\PassOptionsToPackage{noend}{algorithmic}
\documentclass[lettersize,journal]{IEEEtran}
\usepackage{amsmath,amsfonts,amssymb}
\usepackage{algorithmic}
\usepackage{algorithm}
\usepackage{array}
\usepackage{textcomp}
\usepackage{stfloats}
\usepackage{url}
\usepackage{verbatim}
\usepackage{graphicx}
\usepackage{cite}
\usepackage{xcolor}
\usepackage{enumitem}
\usepackage{amsthm}
\usepackage{multicol}
\usepackage{calc}
\usepackage{epsfig}
\usepackage{subcaption}
\usepackage[font=footnotesize]{caption}

\theoremstyle{plain}
\newtheorem{theorem}{Theorem}

\newtheorem{corollary}[theorem]{Corollary}

\theoremstyle{definition}

\newtheorem{remark}{Remark}

\newcommand{\cC}{\mathcal{C}}
\newcommand{\cN}{\mathcal{N}}

\newcommand{\cR}{\mathcal{R}}
\newcommand{\sfH}{\mathsf{H}}
\newcommand{\sfT}{\mathsf{T}}
\newcommand{\bC}{\mathbb{C}}
\newcommand{\bR}{\mathbb{R}}

\newcommand{\trace}{\mathrm{tr}}

\newcommand{\abs}[1]{{\left\lvert #1 \right\rvert}}
\newcommand{\bd}[1]{\boldsymbol{#1}}

\newcommand*\diff{d}
\newcommand{\cond}{\mathchoice{\,\vert\,}{\mspace{2mu}\vert\mspace{2mu}}{\vert}{\vert}}
\def\rank{\mathop{\rm rank}\nolimits}%
\mathchardef\Re="023C
\mathchardef\Im="023D

\DeclareMathOperator\E{\mathbb{E}}

\newcommand\blfootnote[1]{%
\begingroup
\renewcommand\thefootnote{}\footnote{#1}%
\addtocounter{footnote}{-1}%
\endgroup
}

\begin{document}

\title{Active MIMO Sensing: Optimality Conditions and Exploration-Exploitation Tradeoff}

\author{Nadim Ghaddar,~\IEEEmembership{Member,~IEEE,} Kareem M. Attiah,~\IEEEmembership{Member,~IEEE,} and Wei Yu,~\IEEEmembership{Fellow,~IEEE}}


\maketitle

\begin{abstract}
\blfootnote{Manuscript submitted to IEEE Transactions on Information Theory on April 19, 2026, revised on October 24, 2026. Nadim Ghaddar is with the Electrical and Computer Engineering Department, University of British Columbia, BC V6T 1Z4, Canada. Kareem M. Attiah and Wei Yu are with the Electrical and Computer Engineering Department, University of Toronto, Toronto, ON M5S 3G4, Canada. Emails: nadim.ghaddar@ece.ubc.ca, kareem.attiah@mail.utoronto.ca, weiyu@ece.utoronto.ca. Materials in this paper have been presented in part at the IEEE International Conference on Communications (ICC), Montreal, Quebec, Canada, June 2025~\cite{Ghaddar2025ICC}. This work was supported by Natural Science and Engineering Research Council (NSERC) via a Discovery Grant and via the Canada Research Chairs program.}
	This paper develops an active sensing framework for designing the transmit and receive beamformers of a multiple-input multiple-output (MIMO) radar system with a limited number of radio-frequency (RF) chains. In the proposed framework, the transmit and receive beamformers are adaptively designed in each sensing stage based on the measurements made in the previous sensing stages, and the optimization of the beamformers are based on the minimization of the Bayesian Cram{\'e}r-Rao bound (BCRB) for estimating the unknown sensing parameters. To address the exploration-exploitation tradeoff inherent in such an adaptive procedure, this paper proposes a two-phase design involving two variants of the BCRB optimization problem: an exploration-centric variant that requires multiple orthogonal beamforming directions to be probed in each sensing stage, followed by an exploitation-centric variant that optimizes the per-stage BCRB without such requirement. For both variants, the transmit and receive beamformers are optimized as two subproblems in an alternating fashion in each sensing stage. The alternating optimization is shown to converge to a stationary point of the per-stage BCRB optimization provided that each subproblem is solved to global optimality. A main contribution of this paper is to determine sufficient conditions under which global optimality can be guaranteed---despite that the two subproblems are nonconvex (due to the limited number of RF chains), by studying the structures of their optimal solutions via Lagrangian dual analysis. The conditions pertain to the multiplicity of the eigenvalues of specific direction matrices that can be analytically written in terms of the optimal dual variables.  These conditions further imply the tightness of the semidefinite relaxation of the two subproblems. The analysis of the dual problems provides insights on the design of numerical algorithms for beamforming optimization. Simulation results demonstrate the advantages of the proposed design framework as compared to state-of-the-art adaptive beamforming strategies.
\end{abstract}



\section{Introduction}

Multiple-input multiple-output (MIMO) radar systems probe sensing targets by simultaneously sending multiple orthogonal waveforms, whose reflections are processed through multiple receive antennas. The primary advantage of MIMO radars is their ability to provide higher spatial resolution, better interference suppression and enhanced parameter estimation compared to conventional phased-array radars~\cite{Stoica2008}. Harnessing these benefits necessitates the use of highly-directional signaling (i.e., beamforming) in order to effectively steer the energy in the direction of the targets while nulling the interference from other sources. This need is further compelled by practical hardware constraints that often limit the number of radio frequency (RF) chains at the transmit and receive antenna arrays of the radars, and hence the number of beamformers that can be used.


This paper considers an active sensing process in which a MIMO radar system aims to sense multiple targets by sequentially acquiring measurements across multiple sensing stages. In each sensing stage, the radar system can adaptively design the transmit and receive beamformers based on the available historical measurements obtained from previous sensing stages. After a given number of sensing stages, the radar system uses all the acquired measurements to compute an estimate of the desired parameters. This work focuses on the task of adaptively designing the beamformers in order to estimate the angles and the reflection coefficients corresponding to the sensing targets, which are referred to hereafter as \emph{sensing parameters}.

In each sensing stage, the beamformers are designed by optimizing the Bayesian Cram{\'e}r-Rao bound (BCRB), which is a lower bound on the mean-squared error (MSE) of any estimator of the sensing parameters in the Bayesian setting~\cite{VanTrees1968}. This bound has been extensively used as a performance metric for sensing applications (e.g.,~\cite{Hurtado2008,Huleihel2013,Tabrikian2015,Shuowen2024,Shuowen2024_isit, Kareem2025_2,Kareem2025,Ghaddar2025ICC,Shuowen2026HowManyBeams,Shuowen2026}). The main advantage of adopting a Bayesian approach is that non-Bayesian bounds, in general, may depend on the set of unknown parameters that are to be estimated, in which case the optimal beamformers would also depend on those parameters. Furthermore, in many applications such as target tracking (e.g.,~\cite{Hurtado2008}), some prior statistical information on the sensing parameters may be available, which makes the use of Bayesian bounds natural for such applications. 

This work provides new insights into the BCRB optimization problem in the active sensing setting. In particular, this paper shows that the performance of an active sensing strategy via BCRB optimization is governed by an \emph{exploration-exploitation tradeoff}. This tradeoff is a central theme in sequential decision making problems~\cite{Sutton2018}. In the context of MIMO radar sensing, the exploration-exploitation tradeoff is manifested by the radar's need to balance between \emph{exploring} the environment in order to acquire new information about the targets and between \emph{exploiting} the current knowledge in order to narrow down the search. This paper proposes to approach this tradeoff through a two-phase design: in the first phase, an exploration-centric variant of the BCRB optimization problem is considered over multiple stages, in which the transmit beamformers are required to be orthogonal, and in the second phase, an exploitation-centric variant of the BCRB optimization is considered, in which no such requirement is imposed. The exploration-centric variant ensures that multiple orthogonal beamforming directions are probed, whereas the exploitation-centric variant typically steers beamforming toward a single direction. The results of this paper demonstrate that using an initial exploration phase improves the sensing performance in the active sensing setting, particularly in the low signal-to-noise-ratio (SNR) regime.


Furthermore, this paper provides analytic insights into the single-stage BCRB optimization problem for both the transmit and receive beamformer designs. In general, when the number of RF chains at the transmitter and/or receiver is limited, the minimization of the BCRB metric is a nonconvex optimization problem, which is challenging to solve. By analyzing this problem in the Lagrangian dual domain, this paper derives conditions under which the BCRB optimization can be solved efficiently to global optimality, despite the nonconvex constraints. More specifically, this paper shows that the BCRB optimization problem reduces to that of maximizing the beamforming gain with respect to a specific direction matrix that can be derived analytically in terms of the optimal dual variables. When this direction matrix satisfies certain conditions pertaining to the multiplicity of its eigenvalues, a globally optimal solution of the BCRB optimization problem can be found. The derived conditions provide a \emph{certificate of optimality} for the proposed beamformer design under the BCRB metric. The analysis of the dual problem also gives insights on how many transmit beamformers are needed to achieve the optimal BCRB metric, and more importantly on the design of numerical algorithms for beamformer optimization.

\subsection{Related Work}

The design of beamforming strategies for MIMO radar sensing has been extensively investigated in the literature. A design of the transmit beamformers by optimizing a mutual-information-based metric and the mean-squared error has been considered in~\cite{Yang2007,Lops2007}. In~\cite{Stoica2007}, a design based on approximating a given transmit beampattern is proposed, and it is shown that the desired beampattern should focus in the directions of the sensing targets. The use of the Cram{\'e}r-Rao bound as a metric for waveform optimization in MIMO radar systems has been initially considered in~\cite{Bliss2005} in the case of a single target and later extended to multiple targets in~\cite{Bliss2008}. More recently, the design of MIMO radars that are capable of simultaneously serving communication users within an integrated sensing and communication (ISAC) system has been an active area of research, with a plethora of literature that focuses on beamformer optimization, e.g.,~\cite{FanLiu2020,Yonina2020,FanLiu2022,FanLiu2022_2,FanLiu2022_3,Davidson2023,Yonina2023_2,Han2023,Zhu2023,Shuowen2024,Shuowen2024_isit,Kareem2025,Kareem2025_2,Shuowen2026HowManyBeams,Shuowen2026}. 


Adaptive beamforming strategies have been applied for MIMO radar sensing, especially in the context of cognitive radars~\cite{Haykin2006}, and are shown to significantly improve the sensing performance compared to their non-adaptive counterparts~\cite{Goodman2007, Haykin2008, Hurtado2008,Huleihel2013}. Such adaptive designs are particularly useful for target tracking applications, in which the transmitted waveforms have to be adapted based on the previous information about the moving targets. Common metrics used for adaptive beamformer design include the mean-squared tracking error~\cite{Kershaw1994,Kershaw1997}, the Bayesian Cram{\'e}r-Rao bound (BCRB)~\cite{Hurtado2008,Huleihel2013,Tabrikian2015}, and the mutual information between the target impulse response and the observations~\cite{Sen2010}.



Similar adaptive beamforming techniques have been proposed in the context of wireless communication systems in order to perform sensing-like functionalities (e.g., channel estimation)~\cite{Alkhateeb2014, Javidi2019, Sohrabi2021, Sohrabi2022, Tao2024, Han2025, Ghaddar2025ICC}. In~\cite{Alkhateeb2014, Javidi2019}, the beamforming vectors in each sensing stage are selected from a pre-designed beamforming codebook e.g., based on bisection search over the desired spatial range. However, it has been observed that the performance of codebook-based techniques is limited by the quality of the codebook~\cite{Sohrabi2021}. Alternatively, deep-learning-based solutions for active sensing have been proposed in the literature, e.g.~\cite{Sohrabi2021, Sohrabi2022, Tao2024, Han2025}. Despite being codebook-free, these solutions require training new learning models for each problem setting---ensuring generalizability to varying number of sensing stages or varying channel distributions is in general a challenging task \cite{Ghaddar2025ICC}. 

This paper proposes to jointly design the transmit and receive beamformers of a MIMO radar system by optimizing the BCRB metric. The adaptive design of beamformers using BCRB optimization has been considered in~\cite{Huleihel2013,Tabrikian2015,Shuowen2024,Shuowen2024_isit,Kareem2025_2,Kareem2025,Ghaddar2025ICC,Shuowen2026HowManyBeams,Shuowen2026}. While some of these works are intended for ISAC system models, their results can be specialized to the sensing-only model (i.e., with no communication users) when designated beamformers are used for sensing. These works, however, focus on the design of the transmit beamformers only. Moreover, when specialized to the sensing-only model, these works assume that the number of sensing beamformers to be optimized is equal to the number of transmit antennas. This assumption makes the BCRB optimization problem (over the sensing beamformers) convex. Alternatively, this paper considers the challenging case where the number of beamformers is limited (e.g., due to a limited number of RF chains) and strictly less than the number of antennas, in which case the BCRB optimization problem is nonconvex. Recently, the work~\cite{Shuowen2026} considers the single-stage BCRB optimization problem in an RF-chain-limited ISAC system for the case of a single sensing target and a single parameter to be estimated (i.e., the angle parameter). Instead, this paper considers the general setting of multiple sensing targets, and thus multiple parameters to be estimated, and derives sufficient conditions for the global optimality of the BCRB optimization problem. 

For the problem of designing the transmit beamformers, theoretical bounds on the number of beamformers needed to optimize the BCRB metric are provided in~\cite{Shuowen2024_isit},~\cite{Kareem2025} and~\cite{Shuowen2026HowManyBeams}. The work~\cite{Shuowen2024_isit} shows that, in the case of a single sensing target and a single parameter to be estimated (i.e., the angle parameter), a single transmit beamformer is sufficient to optimize the BCRB metric. In the case of multiple sensing targets, the work in~\cite{Kareem2025} and~\cite{Shuowen2026HowManyBeams} uses a rank reduction method for semidefinite programs~\cite{Palomar2010,Pataki1998} to derive upper bounds on the number of transmit beamformers that are sufficient to optimize the BCRB metric. In this paper, we use an analysis of the Lagrangian dual of the BCRB optimization problem to reveal that the number of transmit beamformers needed is at most equal to the multiplicity of the largest eigenvalue of a specific direction matrix (that is derived in terms of the optimal dual variables). While the exact dependence of this multiplicity on the number of sensing targets remains unresolved, the analysis provided in this paper allows to identify the eigenspace (corresponding to the largest eigenvalue) which contains the optimal set of beamformers, thus reducing the search space.

\subsection{Main Contributions}
This paper develops an active sensing framework for designing the transmit and receive beamformers of a MIMO radar system via BCRB optimization. The main contributions of this paper can be summarized as follows:
\begin{enumerate}
\item To design the beamformers adaptively, this paper optimizes the BCRB metric in each sensing stage. Specifically, the posterior distribution of the sensing parameters is tracked and updated across the sensing stages, and used as the prior distribution for the design of the beamformers in the subsequent sensing stages. To address the exploration-exploitation tradeoff that governs the search over the sensing parameter space, this paper proposes a two-phase design in which two variants of the BCRB optimization problem are solved: an exploration-centric variant that constrains the transmit beamformers to be orthogonal, followed by an exploitation-centric variant that optimizes the BCRB without such constraint. The simulation results reported in this paper demonstrate that the use of an initial exploration phase significantly improves the sensing performance in the low-to-moderate SNR regime.
\item For each of the two variants, the transmit and receive beamformers are optimized via an alternating optimization algorithm. That is, the algorithm alternates between optimizing the transmit and the receive beamformers. The alternating optimization algorithm is guaranteed to converge to a stationary point provided that each subproblem is solved to global optimality.
\item Each optimization subproblem (over the transmit/receive beamformers) is analyzed in the Lagrangian dual domain and is shown to simplify to that of maximizing the beamforming gain with respect to some direction matrix that can be analytically written in terms of the optimal dual variables. Despite the nonconvexity of each optimization problem, this paper derives sufficient conditions that guarantee the global optimality of the proposed solution. The sufficient conditions pertain to the multiplicity of the eigenvalues of the direction matrix. Furthermore, it is shown that these conditions imply that the semidefinite relaxation of each optimization problem is tight. 
\item In case the sufficient conditions do not hold, the analysis of the Lagrangian dual problem allows to identify the eigenspaces (corresponding to the direction matrix) that contain the globally optimal set of beamformers. Hence, the analysis allows to reduce the search space of the optimal beamformers. In the case of optimizing over the transmit beamformers, this analysis further implies that the minimum number of transmit beamformers that can achieve the optimal BCRB metric is at most equal to the multiplicity of the largest eigenvalue of the direction matrix. Collectively, these contributions provide a comprehensive understanding of the optimal solutions of the BCRB optimization problem in the context of MIMO radar sensing.
\end{enumerate}

%
%
%

\subsection{Paper Organization and Notation}
The remaining of the paper is organized as follows. Section~\ref{sec:model} presents the system model considered in this paper and the formulation of the two variants of the BCRB optimization problem. Section~\ref{sec:proposed-P1} addresses the exploitation-centric variant, whereas Section~\ref{sec:proposed-P2} focuses on the exploration-centric variant. For both variants, an alternating optimization algorithm is described, and sufficient conditions for global optimality are derived using Lagrangian duality. Section~\ref{sec:proposed} describes the proposed active sensing strategy by utilizing the framework developed in the previous sections. In Section~\ref{sec:simulations}, the performance of the proposed design is evaluated and compared to existing beamforming strategies. Finally, Section~\ref{sec:conclusion} concludes the paper.

Notation: The following notation will be used. Scalars, vectors and matrices are denoted using lowercase, boldface lowercase and boldface uppercase letters respectively. The space of $m\times n$ real (complex) matrices is denoted by $\mathbb{R}^{m\times n}$ ($\mathbb{C}^{m\times n}$). ${\bf I}_n$ is used to denote the $n\times n$ identity matrix. The notation $(\cdot)^\sfT$, $(\cdot)^\sfH$, $(\cdot)^{-1}$ and $\trace(\cdot)$ is used to represent the transpose, conjugate transpose, inverse and trace of a matrix respectively. For a given square matrix ${\bf A}$, $\mu_i({\bf A})$ denotes the $i$-th largest eigenvalue of ${\bf A}$. The operators $\otimes$ and $*$ denote the Kronecker product and the column-wise Khatri-Rao product of two matrices respectively. The notation $\Re\{\cdot\}$ and $\Im\{\cdot\}$ denote the real and imaginary parts of a complex number. $\mathbb{E}[\cdot]$ denotes the expectation operator of a random variable, and $\mathcal{CN}(\bd{\mu}, {\bf \Sigma})$ denotes a circularly-symmetric complex Gaussian distribution with mean $\bd{\mu}$ and covariance matrix ${\bf \Sigma}$.

\section{System Model and Problem Formulation} \label{sec:model}
\subsection{System Model}
This paper considers a narrowband monostatic MIMO radar consisting of two co-located arrays equipped with $N_T$ transmit antennas and $N_R$ receive antennas that aim to adaptively estimate the sensing parameters of $L$ targets. Further, we consider an RF-chain-limited system where the transmit and receive arrays employ $M_T$ and $M_R$ RF chains respectively, with $M_T < N_T$ and $M_R < N_R$. This means that the radar can use at most $M_T$ transmit beamformers and $M_R$ receive beamformers. 

Sensing takes place over multiple stages.
In the \mbox{$k$-th} sensing stage, the transmitted waveform is constructed as a linear combination of a set of $M_T$ temporally orthogonal waveforms to be transmitted over a pulse duration $T_p$, i.e.,
\begin{equation} \label{eqn:transmit}
	{\bf x}_k(t) = {\bf V}_k{\bf s}_k(t),
\end{equation}
where ${\bf s}_k(t) = \begin{bmatrix}
	s_{k,1}(t) &\cdots & s_{k,M_T}(t)
\end{bmatrix}^\sfT$ are temporally-orthogonal waveforms with $\int_0^{T_p} s_{k,i}(t)s_{k,j}^{*}(t)\diff t = \delta_{ij}$ (e.g., chirp signals with appropriate frequency shifts), and ${\bf V}_k \in \mathbb{C}^{N_T \times M_T}$ is the transmit analog beamforming matrix 
consisting of $M_T$ beamformers. The waveform design given in~(\ref{eqn:transmit}) allows to decouple the temporal and spatial components of the transmitted signal, where the orthogonal signals ${\bf s}_k(t)$ are designed based on their delay-Doppler properties (subject to the orthogonality constraint), while the transmit beamforming matrix ${\bf V}_k$ is designed to achieve a specific beam pattern~\cite{Friedlander2012}. The transmitted waveform ${\bf x}_k(t)$ is subject to a total transmit power constraint $P$. Since the orthogonal waveform vector ${\bf s}_k(t)$ is normalized to have unit power, this implies that the transmit beamforming matrix satisfies
\begin{equation}
	\trace({\bf V}_k{\bf V}_k^\sfH) \leq P.
\end{equation}


The radar system aims to sense, within the coherence time, $L$ point targets situated in the far field. This paper considers the more challenging setting in which all $L$ targets fall within the delay-Doppler resolution of the radar (e.g., vehicles with the same velocity in adjacent lanes, or multiple scatterers from an extended target), so that the radar must rely on spatial processing for target separation. In this case, the target response matrix can be expressed as~\cite[Chapter 4]{Stoica2008},
\begin{equation} \label{eqn:H_theta}
	{\bf H}(\bd{\theta}) = \sum_{i = 1}^{L} \alpha_i {\bf a}_{\rm R}(\phi_i){\bf a}_{\rm T}^\sfH (\phi_i),
\end{equation}
where:
\begin{itemize}
	\item $\alpha_i \sim \cC\cN(0,\sigma_\alpha^2)$ is the reflection coefficient of the $i$-th target, which	contains both the round-trip path-loss and the radar cross-section (RCS) of the target,
	\item $\phi_i$ is the azimuth angle of the $i$-th target relative to the radar, assumed to follow a prior distribution that is approximately uniform over a region of interest $[\phi_{\mathrm{min}}, \phi_{\mathrm{max}}]$, with a smoothly tapered density that scales to zero at the boundaries\footnote{For example, a tapered cosine distribution with a small tapering width can be used. Smooth decay is required so that the prior distribution satisfies regularity conditions for BCRB.},
	\item ${\bf a}_{\rm T}(\cdot)$ and ${\bf a}_{\rm R}(\cdot)$ are the transmit and receive array response vectors respectively, each corresponding to a uniform linear array with half-wavelength antenna spacing\footnote{While uniform linear arrays are considered for simplicity, the results of this paper apply with more generality to arbitrary non-uniform and sparse array configurations.}, i.e.,
	\begin{equation}
		{\bf a}_{\rm S}(\phi) = \begin{bmatrix}
			1 & e^{\imath \pi\sin \phi} & \cdots & e^{\imath (N_{\rm S}-1)\pi\sin \phi}
		\end{bmatrix}^\sfT,
	\end{equation}
	for ${\rm S} \in \{{\rm R}, {\rm T}\}$, where $\imath =\sqrt{-1}$ denotes the complex imaginary unit,
	\item $\bd{\theta} = [\bd{\phi}^\sfT, \Re\{\bd{\alpha}^\sfT\}, \Im\{\bd{\alpha}^\sfT\}]^\sfT \in \bR^{3L}$ is the set of sensing parameters with $\bd{\phi} = [\phi_1, \ldots, \phi_L]^\sfT$ and $\bd{\alpha} =[\alpha_1, \ldots, \alpha_L]^\sfT$.
\end{itemize}
This paper assumes that the number of sensing targets $L$ is known to the radar prior to the parameter estimation phase. In practice, this knowledge can be obtained through an earlier beam sweeping stage.

At the receiving end, the radar employs a receive beamforming matrix (i.e., combiner) consisting of $M_R$ beamformers. Since the $L$ targets have similar delay and Doppler shifts, the baseband received signal in the $k$-th sensing stage can be expressed as
\begin{equation} \label{eqn:receive}
	{\bf y}_k(t) = {\bf W}_k^\sfH {\bf H}(\bd{\theta}) {\bf V}_k{\bf s}_k(t - \tau)e^{j\omega_D t} + {\bf W}_k^\sfH {\bf n}_k(t),
\end{equation}
for $t \in [0,T_p]$, where ${\bf W}_k \in \mathbb{C}^{N_R\times M_R}$ is the receive analog beamforming matrix used in the $k$-th sensing stage, $\tau$ and $\omega_D$ are respectively the propagation delay and the Doppler frequency shift of the echo signals reflected from the targets, and ${\bf n}_k(t) \in \mathbb{C}^{N_R\times 1}$ is a white Gaussian noise vector, which is assumed to be temporally and spatially circularly symmetric complex Gaussian with unit-variance components. 

In radar signal processing, the received waveform ${\bf y}_k(t)$ is typically processed through a bank of two-dimensional delay-Doppler filters, each tuned to a pair of propagation delay and Doppler shift. Since the targets in the model of this paper are assumed to have similar delays and Doppler shifts,
the matched filtering process can be expressed as
\begin{equation} \label{eqn:receive-discrete}
	{\bf Y}_k = \int_{0}^{T_p}{\bf y}_k(t) {\bf s}_k^\sfH(t-\tau)e^{-j\omega_D t} \diff t.
\end{equation}
Note that under the narrowband signaling considered in this paper, any mismatch in delay-Doppler tuning can be absorbed through the unknown reflection coefficients $\bd{\alpha}$. Using the orthogonality of the transmitted waveforms, the received signal in~(\ref{eqn:receive-discrete}) can be written as
\begin{equation} \label{eqn:model}
	{\bf Y}_k = {\bf W}_k^\sfH {\bf H}(\bd{\theta}) {\bf V}_k + {\bf W}_k^\sfH {\bf Z}_k,
\end{equation}
where ${\bf Z}_k \in \mathbb{C}^{N_R \times M_T}$ is a matrix of independent, zero-mean, circularly-symmetric complex Gaussian random variables, each with unit variance. 

Finally, we point out that this paper does not impose any unit modulus constraints on the analog beamforming matrices ${\bf V}_k$ and ${\bf W}_k$. Note that analog beamformers without modulus constraints can be implemented using vector modulators~\cite{Ellinger2001} or double phase shifters~\cite{Letaief2016}.


\subsection{Active Sensing}
This paper considers the active sensing setup of cognitive radars~\cite{Haykin2006}, in which a radar system adaptively interrogates the propagation channel using the available information from previous measurements. More specifically, in this active sensing setting, the radar system adaptively designs the transmit and receive beamforming matrices ${\bf V}_k$ and ${\bf W}_k$ in the $k$-th sensing stage based on the received signals from previous sensing stages, where the end goal is to estimate the set of sensing parameters given by $\bd{\theta}$ after a total of $T$ sensing stages. In other words, the radar system sets the beamforming matrices in the $k$-th sensing stage to
\begin{equation}
({\bf V}_{k}, {\bf W}_{k}) = g_k({\bf Y}_{1:k-1}, {\bf V}_{1:k-1}, {\bf W}_{1:k-1})
\end{equation}
for some function $g_k$. After making $T$ measurements ${\bf Y}_{1:T} = ({\bf Y}_1,\ldots, {\bf Y}_T)$ across all sensing stages, the radar computes an estimate of the sensing parameters $\boldsymbol{\hat \theta}$ as a function of all received measurements and beamforming matrices, 
\begin{equation}
\boldsymbol{\hat\theta} = f({\bf Y}_{1:T}, {\bf V}_{1:T}, {\bf W}_{1:T}),
\end{equation}
for some estimator $f$. Note that when $N_R = M_R$ and ${\bf W}_k = {\bf I}_{N_R}$ $\forall \, k$, the problem boils down to adaptively designing the transmit beamforming matrix ${\bf V}_k$ only, which has been previously considered in~\cite{Huleihel2013}. On the other hand, when $N_T = M_T = 1$, the problem reduces to designing the receive beamforming matrix ${\bf W}_k$ only, which is equivalent to the problem considered in~\cite{Ghaddar2025ICC}. In this paper, we consider the general setting where both transmit and receiver beamformers are to be adaptively and jointly designed for accurate radar sensing. 


The performance of a beamforming strategy is measured by the weighted mean-squared error (WMSE), defined as
\begin{equation}
\mathrm{WMSE} = \sum_{i = 1}^{3L}q_i\mathbb{E}\left[\abs{\theta_i - {\hat\theta}_i}^2 \right], \\
\end{equation}
where $(q_1, \ldots, q_{3L})$ are given non-negative weighting coefficients such that $\sum_{i = 1}^{3L}q_i = 1$, and the expectation is over the distribution of all stochastic parameters of the model, i.e., $\bd{\phi}$, $\boldsymbol{\alpha}$, and ${\bf Z}_{1:T}$. The use of the weighting coefficients allows to weigh the mean-squared error of each parameter in $\bd{\theta}$ differently (e.g., different weights for $\bd{\phi}$ and $\bd{\alpha}$), or to discard some of the sensing parameters if there is no interest in estimating them (i.e., nuisance parameters) by setting their corresponding coefficients to zero.

\subsection{Problem Formulation}
Based on the above performance metric, the active sensing problem can be formulated as:
\begin{subequations} \label{eqn:active-sensing}
\begin{align}
\underset{\{g_k\}_{k=1}^{T}, f}{\text{minimize}} \quad &\sum_{i = 1}^{3L}q_i\mathbb{E}\left[\|\theta_i - {\hat\theta}_i\|^2 \right]\\
\text{subject to} \quad &({\bf V}_k, {\bf W}_k) = g_k({\bf Y}_{1:k-1}, {\bf V}_{1:k-1}, {\bf W}_{1:k-1}) \,\, \forall k,\\
&{\boldsymbol{\hat\theta}} = f({\bf Y}_{1:T}, {\bf V}_{1:T}, {\bf W}_{1:T}).
\end{align}
\end{subequations}
Problem~(\ref{eqn:active-sensing}) is a joint optimization over the adaptive sensing strategies $\{g_k\}_{k=1}^{T}$ and the angle estimator function $f$. In general, solving such an optimization over functional expressions is quite challenging. In practice, the common approach is to choose the beamforming matrices from a pre-designed codebook. For example, some existing strategies employ a hierarchical beamforming codebook~\cite{Alkhateeb2014,Javidi2019}, in which an adaptive bisection search of the targets over the angular space is done. The performance of such schemes is governed by the quality of the hierarchical codebook. Instead, it has been shown that the performance can be improved by a codebook-free design, in which problem~(\ref{eqn:active-sensing}) is solved in a data-driven fashion~\cite{Sohrabi2021,Sohrabi2022}. However, the learning-based solutions require training a new model for each specific set of problem settings (e.g., signal-to-noise ratio, number of sensing stages, desired angle resolution, etc.).


This paper proposes to design a codebook-free beamforming strategy analytically. In particular, we consider a Bayesian formulation in which a prior distribution of $\bd{\theta}$ is assumed at the beginning, and then the prior is updated in each subsequent sensing stage and set equal to the posterior distribution of $\bd{\theta}$ given the previous measurements. Indeed, the posterior distribution of $\bd{\theta}$ is a sufficient statistic for the estimation problem. In this Bayesian framework, we adopt sensing strategies $\{g_k^{\mathrm{BCRB}}\}_{k=1}^T$ that minimize the Bayesian Cram{\'e}r-Rao bound (BCRB) as a function of the beamforming matrices at each sensing stage, then finally a minimum mean-squared error (MMSE) estimator $f^{\mathrm{MMSE}}$ to estimate the sensing parameters.

Unlike the classical CRB~\cite{Cramer1946,Rao1945}, which depends on the unknown parameters $\bd{\theta}$ (e.g.,~\cite{Tabrikian2006,Bliss2008,FanLiu2022_2,FanLiu2022_3,Zhu2023}), the BCRB provides a lower bound on the MSE averaged over the prior distribution of $\bd{\theta}$~\cite{VanTrees1968}. Specifically, when applied to the active sensing problem after observing the first $k-1$ measurements, any estimator $\bd{\hat\theta}_k$ must satisfy 
\begin{equation} \label{eqn:conditional_mse}
\E\left[(\bd{\theta}-\bd{\hat\theta}_k)(\bd{\theta} - \bd{\hat\theta}_k)^\sfT \, \Big| \, {\bf Y}_{1:k-1}\right] \succeq {\bf J}_k^{-1}({\bf V}_k, {\bf W}_k),
\end{equation}
where $\succeq$ denotes inequality with respect to the positive semidefinite (PSD) cone, ${\bf J}_k({\bf V}_k,{\bf W}_k) \in \mathbb{C}^{3L\times 3L}$ is the Bayesian Fisher information matrix (BFIM) computed over the prior distribution of $\bd{\theta}$, and the expectation is taken over the joint distribution of $({\bf Y}_k, \bd{\theta})$ given the measurements ${\bf Y}_{1:k-1}$ observed so far. This gives a lower bound on the WMSE of estimating ${\bd{\theta}}$ in the $k$-th sensing stage as
\begin{equation} \label{eqn:mse_lb}
\mathrm{WMSE}_{k} \, \geq \, \trace\left({\bf Q}{\bf J}_k^{-1}({\bf V}_k,{\bf W}_k)\right),
\end{equation}
where ${\bf Q} = \mathrm{diag}\left(\begin{bmatrix}
q_1, \ldots, q_{3L}
\end{bmatrix}\right)$ is the weighting matrix. We remark that the BFIM ${\bf J}_k({\bf V}_k,{\bf W}_k)$ also depends on the previous measurements ${\bf Y}_{1:k-1}$ through the posterior distribution of $\bd{\theta}$. However, since these measurements are assumed to be known in the $k$-th sensing stage, we do not explicitly show that dependence for notational convenience.

This paper proposes to minimize the lower bound given in~(\ref{eqn:mse_lb}) in each sensing stage. We consider two variants of this minimization. In the first variant, we design the beamforming matrices in the $k$-th stage by solving the following optimization problem:
\begin{subequations} \label{eqn:bcrb}
\begin{align}
\mathsf{(P_1)}: \quad \,\, \underset{\substack{{\bf V}_k \in \mathbb{C}^{N_T\times M_T},\\ {\bf W}_k  \in \mathbb{C}^{N_R\times M_R}}}{\text{minimize}} \quad &\trace\left({\bf Q}{\bf J}_k^{-1}({\bf V}_k,{\bf W}_k)\right)\\
\text{subject to} \quad &\trace({\bf V}_k{\bf V}_k^\sfH) \leq P, \label{eqn:bcrb-1}
\end{align}
\end{subequations}
where the only constraint is that the transmit beamforming matrix must satisfy a per-stage power constraint. 
We remark that this problem formulation implicitly assumes that the posterior distribution of $\bd{\theta}$ satisfies the BCRB regularity conditions in each sensing stage~\cite{VanTrees1968}, in which case there exists a set of transmit and receive beamformers ${\bf V}_k$ and ${\bf W}_k$ such that the BFIM ${\bf J}_k({\bf V}_k,{\bf W}_k)$ is invertible.

As we shall see later, the adaptive sensing strategy that optimizes~(\ref{eqn:bcrb}) in each sensing stage has superior performance in the high SNR regime; however, its performance degrades significantly at low and moderate SNRs. In particular, the solution ${\bf V}_k^*$ of the optimization problem~(\ref{eqn:bcrb}) is typically rank-deficient~\cite{Kareem2025,Shuowen2026HowManyBeams}, yet we observe through simulations that utilizing a full-rank transmit beamforming matrix in an active setting improves the sensing performance in the low to moderate SNR regime\footnote{We acknowledge that another reason for the degradation of the performance in the low to moderate SNR regime is that the BCRB tends to be a looser bound on the MSE in that regime.}. The reason for this phenomenon is that the use of a full-rank transmit beamforming matrix allows for further exploration in the spatial domain, which is particularly useful in the low SNR regime, where the noise dominates the signal. For this reason, we consider a second variant of the BCRB optimization problem:
\begin{subequations} \label{eqn:bcrb-var}
\begin{align}
\mathsf{(P_2)}: \quad \,\, \underset{\substack{{\bf V}_k \in \mathbb{C}^{N_T\times M_T},\\ {\bf W}_k  \in \mathbb{C}^{N_R\times M_R}}}{\text{minimize}} \quad &\trace\left({\bf Q}{\bf J}_k^{-1}({\bf V}_k,{\bf W}_k)\right)\\
\text{subject to} \quad &{\bf V}_k^\sfH{\bf V}_k = \frac{P}{M_T}{\bf I}_{M_T}, \label{eqn:bcrb-var-1}
\end{align}
\end{subequations}
where the constraint~(\ref{eqn:bcrb-var-1}), in addition to satisfying the power constraint, also ensures that the transmit beamformers are mutually orthogonal. Hence, in this variant, the BCRB lower bound is minimized while ensuring that multiple orthogonal directions are explored. Following this exposition, we refer to $\mathsf{(P_1)}$ as the \emph{exploitation-centric} variant of the BCRB optimization problem and $\mathsf{(P_2)}$ as the \emph{exploration-centric} variant of the BCRB optimization problem.

Ultimately, the proposed sensing strategy combines both variants: in the first $T_{\mathrm{explore}}$ sensing stages ($0 \leq T_{\mathrm{explore}} \leq T$), the beamforming matrices are designed to optimize $\mathsf{(P_2)}$, thus allowing for some \emph{exploration} in the spatial domain, and in the remaining $T - T_{\mathrm{explore}}$ sensing stages, the beamformers are designed to optimize $\mathsf{(P_1)}$, thus allowing to \emph{exploit} the accumulated knowledge about the target locations. Indeed, the formulation of the active sensing problem touches upon the well-known exploration-exploitation tradeoff, which is a fundamental issue in sequential decision making~\cite{Sutton2018}. In this context, this tradeoff is reflected by the radar's need to balance between exploring different spatial directions in the search space and between exploiting the accumulated knowledge to locate the targets. Note that if $T_{\mathrm{explore}} = 0$, then the sensing strategy optimizes $\mathsf{(P_1)}$ in all sensing stages, and if $T_{\mathrm{explore}}  = T$, the sensing strategy optimizes $\mathsf{(P_2)}$ in all the sensing stages. Moreover, we note that, when the radar uses a single transmit beamformer (i.e., $M_T = 1$), the two optimization problem $\mathsf{(P_1)}$ and $\mathsf{(P_2)}$ are equivalent.

For both variants, the beamforming matrices ${\bf V}_k$ and ${\bf W}_k$ are adapted based on the previous measurements ${\bf Y}_{1:k-1}$. After the receiver makes a new measurement ${\bf Y}_k$ in the $k$-th sensing stage, the posterior distribution of $\bd{\theta}$ given ${\bf Y}_{1:k}$ is computed and used as the prior distribution for the design of $\mathbf{V}_{k+1}$ and $\mathbf{W}_{k+1}$ in the next sensing stage. After $T$ sensing stages, an estimate $\bd{\hat\theta}$ of the sensing parameters is computed via the MMSE estimator, i.e.,
\begin{equation} \label{eqn:estimator}
\bd{\hat\theta} = \E\left[\bd{\theta} \cond {\bf Y}_{1:T}\right].
\end{equation}
Note that the MMSE estimator minimizes the WMSE for any set of non-negative weighting coefficients.

In addition to the BCRB, we remark that there exist other bounds on the mean-squared error that have been considered in the literature, e.g., the Ziv-Zakai bound~\cite{Ziv1969}, the Weiss-Weinstein bound~\cite{Weinstein1985}, and the Reuven-Messer bound~\cite{Reuven1997}, which may be tighter bounds compared to the BCRB. Even within the Cram{\'e}r-Rao bound family, the Miller-Chang bound~\cite{Miller1978} is reported to be tighter in certain scenarios. However, these bounds take complicated forms and are significantly more difficult to optimize. In this paper, we restrict our attention to the BCRB given in~(\ref{eqn:mse_lb}).

In the following two sections, we consider the exploitation-centric and the exploration-centric variants of the BCRB optimization problem separately. We reformulate the two problems in the covariance domain and approach each one using an alternating optimization algorithm based on Lagrangian duality. We begin the exposition with the exploitation-centric variant in the next section.

\section{Exploitation-Centric BCRB Optimization} \label{sec:proposed-P1}
In this section, we focus on the exploitation-centric variant $\mathsf{(P_1)}$ of the BCRB optimization problem, which is given in~(\ref{eqn:bcrb}). This variant of the problem optimizes the BCRB, where the only constraint on the transmit beamformers is that they satisfy the power constraint. This version of the problem is the one typically considered in the literature (e.g.,~\cite{Huleihel2013, Tabrikian2015, Shuowen2024_isit, Kareem2025, Kareem2025_2, Shuowen2024, Shuowen2026HowManyBeams, Shuowen2026,Ghaddar2025ICC}). This setting is exploitative because it does not impose additional constraints on the transmit beamforming directions that are probed (e.g., the optimal solution of $\mathsf{(P_1)}$ may consist of just a single transmit beamformer).

In the following, we provide a comprehensive characterization of the optimal solution of the optimization problem $\mathsf{(P_1)}$. First, we show that the problem can be equivalently written in the covariance domain as a rank-constrained semidefinite program (SDP). Given this reformulation, the problem is approached via an alternating optimization algorithm which alternates between solving for the transmit beamformers and solving for the receive beamformers, where each subproblem is also a rank-constrained SDP. We show that this algorithm converges to a stationary point, provided that each subproblem is solved to global optimality. By analyzing each subproblem in the Lagrangian dual domain, we provide sufficient conditions under which global optimality is guaranteed. Under these conditions, it is also shown that the semidefinite relaxation of each subproblem is tight. 

\subsection{Derivation of BFIM and Problem Reformulation}
The BCRB has been used in prior work as a metric for designing the transmit beamformers assuming a fully-digital receiver architecture (e.g.,~\cite{Huleihel2013}) and for designing the receive beamformers assuming a single-antenna transmitter (e.g.,~\cite{Ghaddar2025ICC}). Here, we derive the BCRB in terms of both the transmit and receive beamformers. Following similar derivations as in~\cite[Appendices A and B]{Huleihel2013}, the entries of the BFIM for the radar channel model~(\ref{eqn:model}) can be expressed as
\begin{align} \label{eqn:fisher}
\left[{\bf J}_k({\bf V}_k,{\bf W}_k)\right]_{i,j} &= -\E\left[\frac{\partial^2\log f({\bf Y}_k, \bd{\theta}\cond  {\bf Y}_{1:k-1})}{\partial\theta_i \partial\theta_j} \Big| {\bf Y}_{1:k-1}\right] \nonumber \\
&\triangleq \left[{\bf J}_k^{({\mathsf D})}({\bf V}_k,{\bf W}_k)\right]_{i,j} + \left[{\bf J}_{k-1}^{({\mathsf P})}\right]_{i,j},
\end{align}
where ${\bf J}_k^{(\mathsf{D})}({\bf V}_k,{\bf W}_k)$ is the part of the BFIM corresponding to the data measurement model,
\begin{align}\label{eqn:fisher_data}
&\left[{\bf J}_k^{({\mathsf D})}({\bf V}_k,{\bf W}_k)\right]_{i,j} = -\E\left[\frac{\partial^2\log f({\bf Y}_k\cond \bd{\theta}, {\bf Y}_{1:k-1})}{\partial\theta_i \partial\theta_j}\Big| {\bf Y}_{1:k-1}\right] \nonumber \\
&= 2\Re\left\{\trace\left(\E\left[{\bf R}_{{\bf W}_k}\dot{{\bf H}}_i(\bd{\theta}) {\bf R}_{{\bf V}_k}\dot{{\bf H}}_j^{\mathsf H}(\bd{\theta})\Big| {\bf Y}_{1:k-1}\right]\right) \right\},
\end{align}
where
\begin{align}
\dot{{\bf H}}_{i}(\bd{\theta}) &= \frac{\partial {\bf H}(\bd{\theta})}{\partial \theta_i}, \label{eqn:channel-derivative-1}\\
{\bf R}_{{\bf W}_k} &= {\bf W}_k({\bf W}_k^{\mathsf H}{\bf W}_k)^{-1}{\bf W}_k^{\mathsf H}, \label{eqn:R_W}\\
{\bf R}_{{\bf V}_k} &= {\bf V}_k{\bf V}_k^{\mathsf H}.
\end{align}
The important observation here is that ${\bf R}_{{\bf W}_k}$ is the orthogonal projection matrix on the range space of ${\bf W}_k$. 

On the other hand, ${\bf J}_{k-1}^{({\mathsf P})}$ is the part of the BFIM corresponding to the posterior distribution, expressed analytically in~(\ref{eqn:fisher_prior}) at the bottom of the page, in which we use the notation 
\begin{figure*}[b]
\hrulefill
\begin{align} 
\left[{\bf J}_{k-1}^{({\mathsf P})}\right]_{i,j} &= -\E\left[\frac{\partial^2\log f(\bd{\theta}\cond  {\bf Y}_{1:k-1})}{\partial\theta_i \partial\theta_j}\Big| {\bf Y}_{1:k-1}\right]  \nonumber \\
&=2\sum_{\ell=1}^{k-1}\Re\left\{\trace\left(\E\left[{\bf R}_{{\bf W}_\ell}\ddot{{\bf H}}_{ij}(\bd{\theta}){\bf R}_{{\bf V}_\ell}{\bf H}^{\mathsf H}(\bd{\theta})\Big| {\bf Y}_{1:k-1}\right]\right)\right\} +2\sum_{\ell=1}^{k-1}\Re\left\{\trace\left(\E\left[{\bf R}_{{\bf W}_\ell}\dot{{\bf H}}_{i}(\bd{\theta}){\bf R}_{{\bf V}_\ell}\dot{{\bf H}}_{j}^{\mathsf H}(\bd{\theta})\Big| {\bf Y}_{1:k-1}\right]\right)\right\} \nonumber \\
&\quad -2\sum_{\ell=1}^{k-1}\Re\left\{\trace\left({\bf Y}_\ell^{\mathsf H}({\bf W}_\ell^{\mathsf H}{\bf W}_\ell)^{-1}{\bf W}_\ell^{\mathsf H} \E\left[\ddot{{\bf H}}_{ij}(\bd{\theta})\Big| {\bf Y}_{1:k-1}\right]{\bf V}_\ell\right)\right\} -\E\left[\frac{\partial^2\log f(\bd{\theta})}{\partial\theta_i \partial\theta_j}\Big| {\bf Y}_{1:k-1}\right], \label{eqn:fisher_prior}
\end{align}
\end{figure*}
\begin{equation} \label{eqn:channel-derivative-2}
\ddot{{\bf H}}_{ij}(\bd{\theta}) = \frac{\partial^2 {\bf H}(\bd{\theta})}{\partial \theta_i\partial \theta_j}.
\end{equation}
Note that ${\bf J}_{k-1}^{({\mathsf P})}$ does not depend on $({\bf V}_k,{\bf W}_k)$, and hence, can be viewed as a constant matrix in the optimization problem~(\ref{eqn:bcrb}). 
In~(\ref{eqn:fisher_data}) and~(\ref{eqn:fisher_prior}), the expectations are taken over the posterior distribution of $\bd{\theta}$ given the previous measurements ${\bf Y}_{1:k-1}$. 
For completeness, the derivations of equations~(\ref{eqn:fisher_data}) and~(\ref{eqn:fisher_prior}) are provided in Appendix~\ref{appx:BFIM}. 

We now make some observations about the optimization problem~(\ref{eqn:bcrb}). First, the BFIM depends on $({\bf V}_k, {\bf W}_k)$ only through ${\bf R}_{{\bf V}_k}$ and ${\bf R}_{{\bf W}_k}$. In fact, when expressed as ${\bf J}_k({\bf R}_{{\bf V}_k},{\bf R}_{{\bf W}_k})$, we observe that the BFIM is a bi-affine function of its arguments. In particular, when ${\bf R}_{{\bf V}_k}$ is held constant, the function $\tilde{{\bf J}}_{{\bf R}_{{\bf V}_k}}^{(1)}({\bf R}_{{\bf W}_k}) \triangleq {\bf J}_k({\bf R}_{{\bf V}_k},{\bf R}_{{\bf W}_k})$ is an affine function of ${\bf R}_{{\bf W}_k}$, and when ${\bf R}_{{\bf W}_k}$ is held constant, the function $\tilde{{\bf J}}_{{\bf R}_{{\bf W}_k}}^{(2)}({\bf R}_{{\bf V}_k}) \triangleq {\bf J}_k({\bf R}_{{\bf V}_k},{\bf R}_{{\bf W}_k})$ is an affine function of ${\bf R}_{{\bf V}_k}$. This allows to rewrite the optimization problem~(\ref{eqn:bcrb}) in terms of ${\bf R}_{{\bf V}_k}$ and ${\bf R}_{{\bf W}_k}$ as
\begin{subequations} \label{eqn:bcrb-corr}
\begin{align}
\underset{{\bf R}_{{\bf V}_k},{\bf R}_{{\bf W}_k}}{\text{minimize}} \quad &\trace\left({\bf Q}{\bf J}_k^{-1}({\bf R}_{{\bf V}_k},{\bf R}_{{\bf W}_k})\right)\\
\text{subject to} \quad &\trace({\bf R}_{{\bf V}_k}) \leq P, \label{eqn:bcrb-corr-1}\\
&\rank({\bf R}_{{\bf V}_k}) \leq M_T, \,\, {\bf R}_{{\bf V}_k} \succeq 0, \label{eqn:bcrb-corr-2}\\
&{\bf R}_{{\bf W}_k} \text{ is an orthogonal projection matrix}, \label{eqn:bcrb-corr-3}\\
&\rank({\bf R}_{{\bf W}_k}) = M_R, \,\, {\bf R}_{{\bf W}_k} \succeq 0, \label{eqn:bcrb-corr-4}
\end{align}
\end{subequations}
where the constraints~(\ref{eqn:bcrb-corr-1}) and~(\ref{eqn:bcrb-corr-2}) are equivalent to being able to construct $M_T$ beamformers ${\bf V}_k$ that satisfy (\ref{eqn:bcrb-1}), and the constraints~(\ref{eqn:bcrb-corr-3}) and~(\ref{eqn:bcrb-corr-4}) are equivalent to being able to construct $M_R$ beamformers ${\bf W}_k$ such that ${\bf R}_{{\bf W}_k} = {\bf W}_k({\bf W}_k^{\mathsf H}{\bf W}_k)^{-1}{\bf W}_k^{\mathsf H}$. 
The rank constraints on ${\bf R}_{{\bf V}_k}$ and  ${\bf R}_{{\bf W}_k}$ are a consequence of the limited number of RF chains in the system model.

To see how to recover ${\bf W}_k$ from an ${\bf R}_{{\bf W}_k}$ that satisfies (\ref{eqn:bcrb-corr-3}) and~(\ref{eqn:bcrb-corr-4}), we note that any rank-$M_R$ orthogonal projection matrix ${\bf R}_{{\bf W}_k}$ can be written as ${\bf R}_{{\bf W}_k} = {\bf UU}^\sfH$, where ${\bf U} \in \mathbb{C}^{N_R \times M_R}$ has orthonormal columns. Thus, choosing ${\bf W}_k = {\bf U}{\bf \Omega}$ for any invertible ${\bf \Omega} \in \mathbb{C}^{M_R\times M_R}$ recovers a feasible beamforming matrix. Note that the orthogonal projection constraint~(\ref{eqn:bcrb-corr-3}) appears only in the optimization over the receive beamformers~\cite{Ghaddar2025ICC}. Unlike the transmit beamformers, the receive beamformers act on both the signal and the noise, which makes the BCRB objective invariant to any invertible transformation ${\bf W}_k \mapsto {\bf W}_k{\bf \Omega}$ and thus dependent only on the \emph{subspace} spanned by ${\bf W}_k$. 

Second, the minimization of the objective function in (\ref{eqn:bcrb-corr}) has a convenient reformulation using the following Schur complement relation
\begin{equation} \label{eqn:schur}
\begin{bmatrix}
{\bf J}_k & \sqrt{q_\ell}{\bf e}_\ell \\
\sqrt{q_\ell}{\bf e}_\ell^\sfT & d_\ell
\end{bmatrix} \succeq 0 \quad \Leftrightarrow \quad d_\ell \geq q_\ell{\bf e}_\ell^\sfT {\bf J}_k^{-1}{\bf e}_\ell \geq 0,
\end{equation}
where ${\bf e}_\ell$ denotes the $\ell$-th column of the $3L\times 3L$ identity matrix. In particular, using the relation~(\ref{eqn:schur}), we can write $\trace\left({\bf Q}{\bf J}_k^{-1}\right)$ as the following minimization problem:
\begin{subequations}
\begin{align}
\trace\left({\bf Q}{\bf J}_k^{-1}\right) = \underset{d_1, \ldots, d_{3L}}{\text{minimize}} \quad &\sum_{\ell = 1}^{3L}d_\ell\\
\text{subject to} \quad &\begin{bmatrix}
	{\bf J}_k &  \sqrt{q_\ell}{\bf e}_\ell \\
	\sqrt{q_\ell}{\bf e}_\ell^\sfT &  d_\ell
\end{bmatrix} \succeq 0, \quad \forall \ell.
\end{align}
\end{subequations}
Hence, the optimization problem~(\ref{eqn:bcrb-corr}) can be rewritten as the following rank-constrained SDP:
\begin{subequations} \label{eqn:bcrb-corr-sdp}
\begin{align}
\underset{\substack{{\bf R}_{{\bf V}_k},{\bf R}_{{\bf W}_k}, \\ d_1, \ldots, d_{3L}}}{\text{minimize}} \quad &\sum_{\ell = 1}^{3L} d_\ell\\
\text{subject to} \quad &\begin{bmatrix}
	{\bf J}_k({\bf R}_{{\bf V}_k},{\bf R}_{{\bf W}_k}) &  \sqrt{q_\ell}{\bf e}_\ell \\
	\sqrt{q_\ell}{\bf e}_\ell^\sfT &  d_\ell
\end{bmatrix} \succeq 0, \quad \forall \ell, \label{eqn:bcrb-corr-sdp-1}\\
&\trace({\bf R}_{{\bf V}_k}) \leq P,\\
&\rank({\bf R}_{{\bf V}_k}) \leq M_T,\,\, {\bf R}_{{\bf V}_k} \succeq 0, \\
&{\bf R}_{{\bf W}_k} \text{ is an orthogonal projection matrix}, \label{eqn:bcrb-corr-sdp-3}\\
&\rank({\bf R}_{{\bf W}_k}) = M_R, \,\, {\bf R}_{{\bf W}_k} \succeq 0.
\end{align}
\end{subequations}
This reformulation has been extensively used in MIMO radar sensing (e.g.,~\cite{Bliss2008,Huleihel2013,Tabrikian2015}), and is crucial for obtaining an efficient numerical algorithm for solving~(\ref{eqn:bcrb}). 

It is worthwhile to point out some key distinctions between~(\ref{eqn:bcrb-corr-sdp}) and other BCRB-based optimization problems considered in the literature. For example, in the sensing-only literature (e.g., \cite{Huleihel2013,Tabrikian2015}), the goal is typically to optimize the transmitted waveform, and hence, it is common to assume a fully-digital receiver (i.e., $M_R = N_R$), so that the BCRB optimization is only over the transmit beamformers (i.e., only over ${\bf R}_{{\bf V}_k}$). Moreover, the works in~\cite{Huleihel2013,Tabrikian2015} implicitly assume that the number of beamformers to be optimized is equal to the number of antennas (i.e., $M_T = N_T$), in which case the BCRB optimization does not have a rank constraint on ${\bf R}_{{\bf V}_k}$, thus making the optimization problem over ${\bf R}_{{\bf V}_k}$ convex. 
In contrast, this paper focuses on systems with limited number of RF chains, for which the rank constraints on 
${\bf R}_{{\bf W}_k}$ and ${\bf R}_{{\bf V}_k}$ make the optimization problem nonconvex.

Although nonconvex versions of the BCRB optimization problem have appeared in previous ISAC literature (e.g.,~\cite{FanLiu2022_2,Zhu2023,Shuowen2024,Kareem2025_2,FanLiu2022_3}), the nonconvexity in these formulations is only due to the rank-1 constraints that are imposed on the covariance matrices corresponding to the communication beamformers (while a full set of $M_T=N_T$ beamformers are used for sensing). Hence, when specialized to the sensing-only model (i.e., with no communication users), the optimization problems in these works become convex.

ISAC systems that use a limited number of sensing beamformers have been considered in~\cite{Kareem2025,Shuowen2026HowManyBeams}. Specifically, the goal of~\cite{Kareem2025,Shuowen2026HowManyBeams} is to derive upper bounds on the minimum number of sensing beamformers that are needed to optimize the BCRB metric. Since this number is to be minimized, the BCRB optimization problem in this case imposes a rank constraint on ${\bf R}_{{\bf V}_k}$, even when specialized to the sensing-only model. In~\cite{Kareem2025,Shuowen2026HowManyBeams}, the semidefinite relaxation method is adopted, in which the optimization problem is first solved without the rank constraint on ${\bf R}_{{\bf V}_k}$, and then a rank reduction method~\cite{Palomar2010,Pataki1998} is used to find a lower-rank solution with the same objective value. We note that the optimization considered in~\cite{Kareem2025,Shuowen2026HowManyBeams} is only over the transmit beamformers.

In contract to \cite{Kareem2025,Shuowen2026HowManyBeams}. this paper focuses on the nonconvex optimization problem~(\ref{eqn:bcrb-corr-sdp}), in which both the transmit and receive beamformers are to be optimized. We approach this problem via an alternating optimization algorithm that alternates between solving for ${\bf R}_{{\bf V}_k}$ and solving for ${\bf R}_{{\bf W}_k}$. For each subproblem, instead of semidefinite relaxation and rank reduction, we characterize the structure of the optimal solution 
and derive sufficient conditions for global optimality
by analyzing the problem in the Lagrangian dual domain.

\subsection{Alternating BCRB Optimization} \label{sec:alternating}
Motivated by the bi-affine structure of ${\bf J}_k({\bf R}_{{\bf V}_k},{\bf R}_{{\bf W}_k})$, this paper proposes to optimize the transmit and receive matrices ${\bf R}_{{\bf V}_k}$ and ${\bf R}_{{\bf W}_k}$ in~(\ref{eqn:bcrb-corr-sdp}) via an alternating optimization algorithm. That is, the algorithm iterates between solving for ${\bf R}_{{\bf W}_k}$ for a fixed ${\bf R}_{{\bf V}_k}$, and solving for ${\bf R}_{{\bf V}_k}$ for a fixed ${\bf R}_{{\bf W}_k}$. Specifically, for a fixed ${\bf R}_{{\bf V}_k}$, the optimization problem over ${\bf R}_{{\bf W}_k}$ can be written as
\begin{subequations} \label{eqn:bcrb-corr-W}
\begin{align}
\underset{{\bf R}_{{\bf W}_k}, d_1, \ldots, d_{3L}}{\text{minimize}} \quad &\sum_{\ell = 1}^{3L} d_\ell\\
\text{subject to} \,\,\,\,\quad &\begin{bmatrix}
	{\bf J}_k({\bf R}_{{\bf V}_k},{\bf R}_{{\bf W}_k}) & \sqrt{q_\ell}{\bf e}_\ell \\
	\sqrt{q_\ell}{\bf e}_\ell^\sfT &  d_\ell
\end{bmatrix} \succeq 0, \,\, \forall \ell,\label{eqn:bcrb-corr-W-1}\\
&{\bf R}_{{\bf W}_k} \text{ is an orthogonal projection matrix},\label{eqn:bcrb-corr-W-2}\\
&\rank({\bf R}_{{\bf W}_k}) = M_R, \,\, {\bf R}_{{\bf W}_k} \succeq 0, \label{eqn:bcrb-corr-W-3}
\end{align}
\end{subequations}
and for a fixed ${\bf R}_{{\bf W}_k}$, the optimization problem over ${\bf R}_{{\bf V}_k}$ can be written as
\begin{subequations} \label{eqn:bcrb-corr-V}
\begin{align}
\underset{{\bf R}_{{\bf V}_k}, d_1, \ldots, d_{3L}}{\text{minimize}} \quad &\sum_{\ell = 1}^{3L} d_\ell\\
\text{subject to} \,\,\,\,\quad &\begin{bmatrix}
	{\bf J}_k({\bf R}_{{\bf V}_k},{\bf R}_{{\bf W}_k}) & \sqrt{q_\ell}{\bf e}_\ell \\
	\sqrt{q_\ell}{\bf e}_\ell^\sfT &  d_\ell
\end{bmatrix} \succeq 0, \,\, \forall \ell,\\
&\trace({\bf R}_{{\bf V}_k}) \leq P,\label{eqn:bcrb-corr-V-2} \\
&\rank({\bf R}_{{\bf V}_k}) \leq M_T,\,\, {\bf R}_{{\bf V}_k} \succeq 0. \label{eqn:bcrb-corr-V-3}
\end{align}
\end{subequations}
We propose to solve~(\ref{eqn:bcrb-corr-sdp}) using the following iterative algorithm:
\begin{enumerate}[leftmargin=1.5em]
\item Initialize a random ${\bf V}_{k}^{(0)}$ such that $\trace\left({\bf V}_k^{(0)}({\bf V}_k^{(0)})^{\mathsf H}\right) \leq P$.
\item For iteration index $i =1,2,\ldots, I_{\mathrm{max}}$, do:
\begin{enumerate}[label=\roman*)]
\item Find the solution ${\bf W}_k^{(i)}$ of~(\ref{eqn:bcrb-corr-W}) when ${\bf R}_{{\bf V}_k} = {\bf R}_{{\bf V}_k^{(i-1)}}$.
\item Find the solution ${\bf V}_k^{(i)}$ of~(\ref{eqn:bcrb-corr-V}) when ${\bf R}_{{\bf W}_k} = {\bf R}_{{\bf W}_k^{(i)}}$.
\end{enumerate}
\end{enumerate}
This algorithm has the property that the BCRB metric does not increase across iterations. In particular, if we denote
\begin{equation}
\beta_i \triangleq \trace\left({\bf Q}{\bf J}_k^{-1}\big({\bf R}_{{\bf V}_k^{(i)}},{\bf R}_{{\bf W}_k^{(i)}}\big)\right)
\end{equation}
as the value of the objective function at the end of the $i$-th iteration, then the sequence $\{\beta_i\}_{i=1}^{I_{\mathrm{max}}}$ is guaranteed to be non-negative and non-increasing. This implies that the alternating optimization algorithm converges to a limit point as \mbox{$I_{\mathrm{max}} \to \infty$}. As discussed later in Remark~\ref{remark:alt-opt-convergence}, under certain conditions, every such limit point is a stationary point of the optimization problem~(\ref{eqn:bcrb-corr}). 

The optimization problems~(\ref{eqn:bcrb-corr-W}) and~(\ref{eqn:bcrb-corr-V}) are nonconvex due to the orthogonal projection constraint~(\ref{eqn:bcrb-corr-W-2}) and the rank constraints~(\ref{eqn:bcrb-corr-W-3}) and~(\ref{eqn:bcrb-corr-V-3}). In the following, we show that despite the nonconvex constraints, the optimization problems~(\ref{eqn:bcrb-corr-W}) and~(\ref{eqn:bcrb-corr-V}) can be solved to global optimality, provided that certain sufficient conditions hold. To derive these conditions, we approach the two optimization problems in the Lagrangian dual domain.

\subsection{Receive Beamforming Via Lagrangian Duality} \label{sec:dual-W}
In this subsection, we focus on the optimization of receive beamformers in problem~(\ref{eqn:bcrb-corr-W}). We derive its dual problem, and show that the semidefinite relaxation of~(\ref{eqn:bcrb-corr-W}) is tight, under a specific multiplicity condition. In particular, we prove the following theorem.
\begin{theorem} \label{thm:dual-W}
	The dual problem of~(\ref{eqn:bcrb-corr-W}) can be expressed as
	\begin{equation} \label{eqn:bcrb-dual-W}
		\underset{{\bf \Lambda} \in \bC^{3L\times 3L}}{\mathrm{maximize}} \quad 2\trace({\bf \Lambda}{\bf Q}^{1/2}) - \trace\left({\bf \Lambda}^{\mathsf T}{\bf J}_{k-1}^{({\mathsf P})}{\bf \Lambda}\right) - 2\sum_{i=1}^{M_R} \mu_{i}\left({\bf P}_{{\bf \Lambda}  {\bf \Lambda}^{\mathsf T}}\right),
	\end{equation}
	where $\mu_i(\cdot)$ denotes the $i$-th largest eigenvalue, and
	\begin{equation} \label{eqn:P_Lambda-W}
		{\bf P}_{{\bf \Lambda}  {\bf \Lambda}^{\mathsf T}} \triangleq \sum_{m=1}^{M_T}\E\left[\dot{{\bf G}}_{k,m}(\bd{\theta}){\bf \Lambda} {\bf \Lambda}^{\mathsf T}\dot{{\bf G}}_{k,m}^{\mathsf H}(\bd{\theta}) \,\big|\, {\bf Y}_{1:k-1}\right],
	\end{equation}
	in which $\dot{{\bf G}}_{k,m}(\bd{\theta}) = \begin{bmatrix}
	\dot{{\bf H}}_{1}(\bd{\theta}){\bf v}_{k,m} & \cdots & \dot{{\bf H}}_{3L}(\bd{\theta}){\bf v}_{k,m}
	\end{bmatrix}$ and ${\bf v}_{k,m}$ is the $m$-th beamformer of ${\bf V}_k$. The dual problem~(\ref{eqn:bcrb-dual-W}) is unconstrained and convex. Let ${\bf \Lambda}^{*}$ denote the solution of~(\ref{eqn:bcrb-dual-W}). If
	\begin{equation} \label{eqn:spectral-W}
		\mu_{M_R}\left({\bf P}_{{\bf \Lambda}^{*}({\bf \Lambda}^{*})^{\mathsf T}}\right) > \mu_{M_R+1}\left({\bf P}_{{\bf \Lambda}^{*}({\bf \Lambda}^{*})^{\mathsf T}}\right),
	\end{equation}
	then
	\begin{equation}  \label{eqn:R_star-W}
		{\bf R}_{{\bf W}_k}^{*} = {\bf W}_k^{*}({\bf W}_k^{*})^{\mathsf H}
	\end{equation}
	is the globally optimal solution of~(\ref{eqn:bcrb-corr-W}), where
	\begin{equation} \label{eqn:W_star}
		{\bf W}_k^{*} = \begin{bmatrix}
			{\bf w}_{k,1}^{*} & \cdots & {\bf w}_{k,M_R}^{*}
		\end{bmatrix}
	\end{equation}
	and ${\bf w}_{k,1}^{*},\ldots, {\bf w}_{k,M_R}^{*}$ are the eigenvectors of ${\bf P}_{{\bf \Lambda}^{*}({\bf \Lambda}^{*})^{\mathsf T}}$ (with unit norm) corresponding to the $M_R$ largest eigenvalues. Moreover, the optimal primal and dual variables are related according to
	\begin{equation} \label{eqn:opt-dual-variable-W_new}
		{\bf \Lambda}^* = {\bf J}_k^{-1}({\bf R}_{{\bf V}_k}, {\bf R}_{{\bf W}_k}^*){\bf Q}^{1/2}.
	\end{equation}
\end{theorem}

Theorem~\ref{thm:dual-W} gives the representation of the optimization problem~(\ref{eqn:bcrb-corr-W}) in the Lagrangian dual domain. It shows that the dual problem is \emph{unconstrained} and \emph{convex}. Moreover, a globally optimal solution of the primal problem~(\ref{eqn:bcrb-corr-W}) can be derived from the optimal solution of the dual problem, provided that the condition~(\ref{eqn:spectral-W}) holds. This condition pertains to the multiplicity of the eigenvalues of ${\bf P}_{{\bf \Lambda}^{*}({\bf \Lambda}^{*})^{\mathsf T}}$: if there is a strictly positive gap between the $M_R$-th and $(M_R+1)$-th eigenvalues (counting multiplicities), then global optimality is guaranteed. Hence, the condition~(\ref{eqn:spectral-W}) serves as a certificate of optimality of the solution given in~(\ref{eqn:R_star-W}). In other words, once a solution ${\bf \Lambda}^*$ of the dual problem~(\ref{eqn:bcrb-dual-W}) is obtained, one can check the global optimality of the solution given in~(\ref{eqn:R_star-W}) by examining the multiplicity of the eigenvalues of ${\bf P}_{{\bf \Lambda}^{*}({\bf \Lambda}^{*})^{\mathsf T}}$. We note that ${\bf \Lambda}^*$ is unique whenever ${\bf J}_{k-1}^{({\mathsf P})}$ is positive definite. This is because the objective of~(\ref{eqn:bcrb-dual-W}) is strictly concave in this case. The proof of Theorem~\ref{thm:dual-W} is provided in Appendix~\ref{appx:dual-proof-W}.

\begin{remark} \label{remark:spectral-W}
The proof of Theorem~\ref{thm:dual-W} implies that the optimal receive beamformers necessarily belong to the eigenspaces corresponding to the $M_R$ largest eigenvalues of ${\bf P}_{{\bf \Lambda}^{*}({\bf \Lambda}^{*})^{\mathsf T}}$. If condition~(\ref{eqn:spectral-W}) holds, these eigenspaces are completely characterized by the optimal solution~(\ref{eqn:R_star-W}), and thus, global optimality is guaranteed. If condition~(\ref{eqn:spectral-W}) does not hold (i.e., $\mu_{M_R}({\bf P}_{{\bf \Lambda}^{*}({\bf \Lambda}^{*})^{\mathsf T}}) = \mu_{M_R+1}\left({\bf P}_{{\bf \Lambda}^{*}({\bf \Lambda}^{*})^{\mathsf T}}\right)$), there is some ambiguity about the optimal choice of the $M_R$-th beamformer. While the optimal choice is guaranteed to belong to the eigenspace corresponding to $\mu_{M_R}({\bf P}_{{\bf \Lambda}^{*}({\bf \Lambda}^{*})^{\mathsf T}})$, the exact optimizer within the eigenspace is not known, in which case global optimality cannot be guaranteed. 
\end{remark}


\begin{corollary} \label{cor:dual-W}
	If condition~(\ref{eqn:spectral-W}) holds, the semidefinite relaxation of the optimization problem~(\ref{eqn:bcrb-corr-W}), as given below, is tight:
	\begin{subequations} \label{eqn:bcrb-corr-W-sdr}
		\vspace*{-1em}
		\begin{align}
			\underset{{\bf R}_{{\bf W}_k}, d_1, \ldots, d_{3L}}{\text{minimize}} \quad &\sum_{\ell = 1}^{3L} d_\ell\\
			\text{subject to} \,\,\,\,\quad &\begin{bmatrix}
				{\bf J}_k({\bf R}_{{\bf V}_k},{\bf R}_{{\bf W}_k}) & \sqrt{q_\ell}{\bf e}_\ell \\
				\sqrt{q_\ell}{\bf e}_\ell^\sfT &  d_\ell
			\end{bmatrix} \succeq 0, \,\, \forall \ell, \label{eqn:bcrb-corr-W-sdr-1}\\
			&{\bf 0} \preceq {\bf R}_{{\bf W}_k} \preceq {\bf I}_{N_R},  \label{eqn:bcrb-corr-W-sdr-2}\\
			&\trace({\bf R}_{{\bf W}_k}) = M_R.  \label{eqn:bcrb-corr-W-sdr-3}
		\end{align}
	\end{subequations}
\end{corollary}

\begin{proof}
The semidefinite relaxation~(\ref{eqn:bcrb-corr-W-sdr}) is a convex optimization problem that satisfies Slater's condition (since the BFIM is invertible for each ${\bf R}_{{\bf W}_k}$). Thus, strong duality holds for~(\ref{eqn:bcrb-corr-W-sdr}). Moreover, following similar steps as in the proof of Theorem~\ref{thm:dual-W}, it can be shown that the dual of the semidefinite relaxation~(\ref{eqn:bcrb-corr-W-sdr}) with respect to the constraint~(\ref{eqn:bcrb-corr-W-sdr-1}) is equivalent to~(\ref{eqn:bcrb-dual-W}). Thus, by strong duality, the optimization problems~(\ref{eqn:bcrb-dual-W}) and~(\ref{eqn:bcrb-corr-W-sdr}) attain the same optimal objective values. On the other hand, we know from Theorem~\ref{thm:dual-W} that when condition~(\ref{eqn:spectral-W}) holds, strong duality is satisfied for the optimization problem~(\ref{eqn:bcrb-corr-W}) and that its dual problem is given by~(\ref{eqn:bcrb-dual-W}). This implies the tightness of the semidefinite relaxation~(\ref{eqn:bcrb-corr-W-sdr}) under condition~(\ref{eqn:spectral-W}).
\end{proof}


Interestingly, the semidefinite relaxation~(\ref{eqn:bcrb-corr-W-sdr}) coincides with the dual of the dual of the primal problem~(\ref{eqn:bcrb-corr-W}). Note that the constraints~(\ref{eqn:bcrb-corr-W-sdr-2}) and~(\ref{eqn:bcrb-corr-W-sdr-3}) in the semidefinite relaxation represent the convex hull of the constraints~(\ref{eqn:bcrb-corr-W-2}) and~(\ref{eqn:bcrb-corr-W-3}) in the primal problem. In fact, the set of matrices satisfying~(\ref{eqn:bcrb-corr-W-2}) and~(\ref{eqn:bcrb-corr-W-3}) are the extreme points of this convex hull. 


%

The structure of the optimal solution characterized in Theorem~\ref{thm:dual-W} and Corollary~\ref{cor:dual-W} suggests three distinct approaches for finding the optimal receive beamforming matrix. All three recover the optimal dual solution ${\bf \Lambda}^*$ of~(\ref{eqn:bcrb-dual-W}), from which the low-rank solution~(\ref{eqn:R_star-W}) can be found through an eigenvalue decomposition of ${\bf P}_{{\bf \Lambda}^*({\bf \Lambda}^*)^\sfT}$. By the uniqueness of ${\bf \Lambda}^*$, the three approaches are guaranteed to return the same beamformers. The main distinction between the three approaches is in how ${\bf \Lambda}^*$ is computed:

\begin{itemize}[leftmargin=1.5em]
	\item[1)] Semidefinite relaxation of the primal: The first approach solves the semidefinite relaxation~(\ref{eqn:bcrb-corr-W-sdr}) and denotes its solution by $\tilde{{\bf R}}_{{\bf W}_k}^*$.  However, $\tilde{{\bf R}}_{{\bf W}_k}^*$ may not be low rank, so it may not be easy to directly recover the optimal beamformers from $\tilde{{\bf R}}_{{\bf W}_k}^*$. 
	On the other hand, Corollary~\ref{cor:dual-W} asserts that strong duality holds for~(\ref{eqn:bcrb-corr-W-sdr}) and that the dual problem is given by~(\ref{eqn:bcrb-dual-W}). Therefore, the optimal solution ${\bf \Lambda}^*$ can be obtained as
	\begin{equation} \label{eqn:opt-dual-variable-W}
		{\bf \Lambda}^* = {\bf J}_k^{-1}({\bf R}_{{\bf V}_k}, \tilde{{\bf R}}_{{\bf W}_k}^*){\bf Q}^{1/2},
	\end{equation}
	irrespective of whether condition~(\ref{eqn:spectral-W}) holds. 
	From the optimal ${\bf \Lambda}^*$, a low-rank solution~(\ref{eqn:R_star-W}) can then be found through an eigenvalue decomposition of ${\bf P}_{{\bf \Lambda}^*({\bf \Lambda}^*)^\sfT}$.
	Note that whether $\tilde{{\bf R}}_{{\bf W}_k}^*$ is low-rank causes no ambiguity in~(\ref{eqn:opt-dual-variable-W}), because by the uniqueness of ${\bf \Lambda}^*$, every optimal solution of~(\ref{eqn:bcrb-corr-W-sdr}) yields the same ${\bf \Lambda}^*$.

	\item[2)] Gradient-descent on the dual: The second approach solves the dual problem~(\ref{eqn:bcrb-dual-W}) directly using a gradient-descent-based algorithm, which is guaranteed to converge to a globally optimal solution since the dual problem~(\ref{eqn:bcrb-dual-W}) is unconstrained and convex. The main advantage of this approach is that the optimization variable ${\bf \Lambda}$ lies in $\mathbb{R}^{3L\times 3L}$, whose size scales with the number of sensing targets $L$ rather than the number of receive antennas $N_R$. Since $L$ is typically much smaller than $N_R$, this is a drastically smaller search space than the $N_R\times N_R$ matrix variable of~(\ref{eqn:bcrb-corr-W-sdr}). We remark that the objective function in~(\ref{eqn:bcrb-dual-W}) is differentiable if and only if there is a strict gap between the $M_R$-th and $(M_R+1)$-th largest eigenvalues of ${\bf P}_{{\bf \Lambda}  {\bf \Lambda}^{\mathsf T}}$. Hence, when this condition does not hold, one must use a subgradient method.
	
	\item[3)] SDP formulation of the dual: It is possible to formulate the dual problem~(\ref{eqn:bcrb-dual-W}) as an SDP. However, such an SDP formulation would involve a semidefinite constraint on a matrix whose dimension is quadratic in the number of receive antennas. This renders the complexity of this approach too high for practical implementation, especially when $N_R$ is large.
\end{itemize}

In this paper, we adopt the first approach, which is the simplest to implement using off-the-shelf convex solvers and is adequate for the antenna array sizes considered in this paper. For systems with a large number of antennas, the second approach would be more scalable.

\subsection{Transmit Beamforming Via Lagrangian Duality} \label{sec:dual-V}
In this subsection, we consider the optimization problem~(\ref{eqn:bcrb-corr-V}) for the design of the transmit beamformers. As in the previous section, we show that the dual problem of~(\ref{eqn:bcrb-corr-V}) can be expressed as an unconstrained convex problem, and we derive a sufficient condition under which the semidefinite relaxation of~(\ref{eqn:bcrb-corr-V}) is tight. The following theorem holds for the optimal design of the transmit beamformers.

\begin{theorem} \label{thm:dual-V}
	The dual problem of~(\ref{eqn:bcrb-corr-V}) can be expressed as
	\begin{equation} \label{eqn:bcrb-dual-V}
		\underset{{\bf \Lambda} \in \bC^{3L\times 3L}}{\mathrm{maximize}} \quad 2\trace({\bf \Lambda}{\bf Q}^{1/2}) - \trace\left({\bf \Lambda}^{\mathsf T}{\bf J}_{k-1}^{({\mathsf P})}{\bf \Lambda}\right) - 2P\mu_{1}\left(\tilde{{\bf P}}_{{\bf \Lambda}  {\bf \Lambda}^{\mathsf T}}\right),
	\end{equation}
	where $\mu_1(\cdot)$ denotes the largest eigenvalue, and
	\begin{equation} \label{eqn:P_Lambda-V}
		\tilde{{\bf P}}_{{\bf \Lambda}  {\bf \Lambda}^{\mathsf T}} \triangleq \sum_{m=1}^{M_R}\E\left[\dot{\tilde{{\bf G}}}_{k,m}(\bd{\theta}){\bf \Lambda} {\bf \Lambda}^{\mathsf T}\dot{\tilde{{\bf G}}}_{k,m}^{\mathsf H}(\bd{\theta}) \,\big|\, {\bf Y}_{1:k-1}\right],
	\end{equation}
	in which $\dot{\tilde{{\bf G}}}_{k,m}(\bd{\theta}) = \begin{bmatrix}
	\dot{{\bf H}}_{1}^\sfH(\bd{\theta}){\bf w}_{k,m} & \cdots & \dot{{\bf H}}_{3L}^\sfH(\bd{\theta}){\bf w}_{k,m}
	\end{bmatrix}$ and ${\bf w}_{k,m}$ is the $m$-th beamformer of ${\bf W}_k$. The dual problem~(\ref{eqn:bcrb-dual-V}) is unconstrained and convex. Let $\tilde{{\bf \Lambda}}^{*}$ denote the solution of~(\ref{eqn:bcrb-dual-V}). If
	\begin{equation} \label{eqn:spectral-V}
		\mu_{1}\left(\tilde{{\bf P}}_{\tilde{{\bf \Lambda}}^{*}(\tilde{{\bf \Lambda}}^{*})^{\mathsf T}}\right) > \mu_{2}\left(\tilde{{\bf P}}_{\tilde{{\bf \Lambda}}^{*}(\tilde{{\bf \Lambda}}^{*})^{\mathsf T}}\right),
	\end{equation}
	then
	\begin{equation}  \label{eqn:R_star-V}
		{\bf R}_{{\bf V}_k}^{*} = {\bf V}_k^{*}({\bf V}_k^{*})^{\mathsf H}
	\end{equation}
	is the globally optimal solution of~(\ref{eqn:bcrb-corr-V}), where
	\begin{equation} \label{eqn:V_star}
		{\bf V}_k^{*} = \sqrt{\frac{P}{M_T}}\begin{bmatrix}
			{\bf v}_{k,1}^{*} & \cdots & {\bf v}_{k,1}^{*}
		\end{bmatrix}
	\end{equation}
	and ${\bf v}_{k,1}^{*}$ is the eigenvector of $\tilde{{\bf P}}_{\tilde{{\bf \Lambda}}^{*}(\tilde{{\bf \Lambda}}^{*})^{\mathsf T}}$ (with unit norm) corresponding to the largest eigenvalue. Moreover, the optimal primal and dual variables are related according to
	\begin{equation} \label{eqn:opt-dual-variable-V_new}
		\tilde{{\bf \Lambda}}^* = {\bf J}_k^{-1}({\bf R}_{{\bf V}_k}^*, {\bf R}_{{\bf W}_k}){\bf Q}^{1/2}.
	\end{equation}
\end{theorem}

\begin{proof}
	The proof of Theorem~\ref{thm:dual-V} follows similar steps as that of Theorem~\ref{thm:dual-W}. A proof sketch is given in Appendix~\ref{appx:dual-proof-V}.
\end{proof}

Notice the difference between the two dual problems~(\ref{eqn:bcrb-dual-W}) and~(\ref{eqn:bcrb-dual-V}): while~(\ref{eqn:bcrb-dual-W}) depends on the $M_R$ largest eigenvalues of a specific direction matrix,~(\ref{eqn:bcrb-dual-V}) depends only on the single largest eigenvalue. Also, by comparing the solutions~(\ref{eqn:W_star}) and~(\ref{eqn:V_star}) of the two problems (which are optimal if the conditions~(\ref{eqn:spectral-W}) and~(\ref{eqn:spectral-V}) hold), one can observe that~(\ref{eqn:W_star}) consists of $M_R$ orthogonal beamformers, whereas~(\ref{eqn:V_star}) consists of a single beamformer. These differences follow from the fact that the choice of the receive beamforming matrix ${\bf W}_k$ in the channel model~(\ref{eqn:model}) affects the noise covariance matrix, which ultimately results in imposing the orthogonal projection constraint~(\ref{eqn:bcrb-corr-W-2}). 


\begin{remark} \label{remark:spectral-V}
	The proof of Theorem~\ref{thm:dual-V} implies that the optimal transmit beamformers necessarily belong to the eigenspace corresponding to the largest eigenvalue of $\tilde{{\bf P}}_{\tilde{{\bf \Lambda}}^{*}(\tilde{{\bf \Lambda}}^{*})^{\mathsf T}}$. If the multiplicity of the largest eigenvalue is one (i.e., condition~(\ref{eqn:spectral-V}) holds), the eigenspace is completely characterized by the unique largest eigenvector (i.e., the optimal solution~(\ref{eqn:V_star})). If the multiplicity of the largest eigenvalue is larger than one (i.e., condition~(\ref{eqn:spectral-V}) does not hold), there is some ambiguity about the optimal choice of the beamformers within the eigenspace, in which case global optimality cannot be guaranteed. 
\end{remark}

\begin{remark} \label{remark:number-of-beamformers}
	The fact that the optimal transmit beamformers belong to the eigenspace corresponding to the largest eigenvalue of $\tilde{{\bf P}}_{\tilde{{\bf \Lambda}}^{*}(\tilde{{\bf \Lambda}}^{*})^{\mathsf T}}$ implies that the minimum number of transmit beamformers that is sufficient to optimize the BCRB metric is at most equal to the multiplicity of the largest eigenvalue. While the exact dependence of the minimum number of beamformers on the system parameters (e.g., the prior distribution, the number of sensing targets, the number of transmit/receive antennas) is complicated~\cite{Kareem2025,Shuowen2026HowManyBeams}, the identification of the optimal eigenspace allows to reduce the search space. We remark that this observation is closely related to~\cite[Proposition 1]{Shuowen2024_globecom}, where the main technical distinction lies in the fact that the transmit beamformers of~\cite{Shuowen2024_globecom} are fully digital, which leads to a convex optimization. In contrast, the approach adopted in this paper is based on the dual of a nonconvex problem. The nonconvexity is due to the limited number of RF chains.
\end{remark}


\begin{corollary} \label{cor:dual-V}
	If condition~(\ref{eqn:spectral-V}) holds, the semidefinite relaxation of the optimization problem~(\ref{eqn:bcrb-corr-V}), as given below, is tight:
	\begin{subequations} \label{eqn:bcrb-corr-V-sdr}
		\vspace*{-1em}
		\begin{align}
			\underset{{\bf R}_{{\bf V}_k}, d_1, \ldots, d_{3L}}{\text{minimize}} \quad &\sum_{\ell = 1}^{3L} d_\ell\\
			\text{subject to} \,\,\,\,\quad &\begin{bmatrix}
				{\bf J}_k({\bf R}_{{\bf V}_k},{\bf R}_{{\bf W}_k}) & \sqrt{q_\ell}{\bf e}_\ell \\
				\sqrt{q_\ell}{\bf e}_\ell^\sfT &  d_\ell
			\end{bmatrix} \succeq 0, \,\, \forall \ell, \label{eqn:bcrb-corr-V-sdr-1}\\
			&\trace({\bf R}_{{\bf V}_k}) \leq P, \\
			&{\bf R}_{{\bf V}_k} \succeq 0.
		\end{align}
	\end{subequations}
\end{corollary}

\begin{proof}
	The semidefinite relaxation~(\ref{eqn:bcrb-corr-V-sdr}) can be obtained by dropping the rank constraint from the optimization problem~(\ref{eqn:bcrb-corr-V}). When condition~(\ref{eqn:spectral-V}) holds, Theorem~\ref{thm:dual-V} asserts that the optimization problem~(\ref{eqn:bcrb-corr-V}) has a rank-1 solution, in which case the semidefinite relaxation~(\ref{eqn:bcrb-corr-V-sdr}) is tight.
\end{proof}

\begin{remark}
	Note that the dual of the semidefinite relaxation~(\ref{eqn:bcrb-corr-V-sdr}) with respect to the constraint~(\ref{eqn:bcrb-corr-V-sdr-1}) is equivalent to the dual problem~(\ref{eqn:bcrb-dual-V}).
\end{remark}

Theorem~\ref{thm:dual-V} and Corollary~\ref{cor:dual-V} suggest that the BCRB optimization of the transmit beamformers can be done in a manner analogous to the optimization of the receive beamformers in Section~\ref{sec:dual-W}. In particular, the same three approaches can be used to compute the optimal solution $\tilde{{\bf \Lambda}}^{*}$ of the dual problem~(\ref{eqn:bcrb-dual-V}). The first approach involves solving the semidefinite relaxation~(\ref{eqn:bcrb-corr-V-sdr}), denote its solution by $\tilde{{\bf R}}_{{\bf V}_k}^*$ and then recovering the dual variable as
\begin{equation} \label{eqn:opt-dual-variable-V}
	\tilde{{\bf \Lambda}}^{*} = {\bf J}_k^{-1}(\tilde{{\bf R}}_{{\bf V}_k}^*, {\bf R}_{{\bf W}_k}){\bf Q}^{1/2}.
\end{equation}
Alternatively, $\tilde{{\bf \Lambda}}^{*}$ can also be obtained by directly solving the dual problem~(\ref{eqn:bcrb-dual-V}) using a gradient-descent-based algorithm or by formulating the dual problem~(\ref{eqn:bcrb-dual-V}) as an SDP. As before, the first approach is the simplest to implement using off-the-shelf convex solvers and is the one adopted in this paper, whereas the second approach is more scalable with the size of the antenna array. In each approach, a low-rank solution can be obtained by taking the eigenvalue decomposition of $\tilde{{\bf P}}_{\tilde{{\bf \Lambda}}^{*}(\tilde{{\bf \Lambda}}^{*})^{\mathsf T}}$ then computing ${\bf R}_{{\bf V}_k}^{*}$ according to~(\ref{eqn:R_star-V}). This solution is optimal if the multiplicity of the largest eigenvalue of $\tilde{{\bf P}}_{\tilde{{\bf \Lambda}}^{*}(\tilde{{\bf \Lambda}}^{*})^{\mathsf T}}$ is one (i.e., condition~(\ref{eqn:spectral-V}) holds). If the multiplicity is larger than one, we can set the transmit beamformer to be a random vector in the corresponding eigenspace, while ensuring that the power constraint is satisfied, to obtain an approximate solution.

\begin{remark} \label{remark:alt-opt-convergence}
	A notable consequence of conditions~(\ref{eqn:spectral-W}) and~(\ref{eqn:spectral-V}) in Theorems~\ref{thm:dual-W} and~\ref{thm:dual-V} is that they guarantee convergence of the alternating optimization algorithm described in Section~\ref{sec:alternating}. Specifically, under these conditions, the minimizers of the optimization problems~(\ref{eqn:bcrb-corr-W}) and~(\ref{eqn:bcrb-corr-V}) are unique. This implies that not only the optimal value of the alternating optimization algorithm converges, but also every limit point of the sequence $\big\{({\bf R}_{{\bf V}_k^{(i)}},{\bf R}_{{\bf W}_k^{(i)}})\big\}$ is a stationary point of the joint problem~\cite[Theorem 5]{Peng2023}.
\end{remark}

\begin{remark}
	Note that~\cite{Kareem2025} and~\cite{Shuowen2026HowManyBeams} provide upper bounds on the minimum number of transmit beamformers needed to achieve the optimal BCRB. These bounds are general in the sense that they are a function of the number of sensing targets only. We remark that if $M_T$ is larger than or equal to these upper bounds, the rank constraint in~(\ref{eqn:bcrb-corr-V}) can be dropped without loss of generality, and the low-rank solution can be obtained by solving the semidefinite relaxation~(\ref{eqn:bcrb-corr-V-sdr}) and then applying the rank reduction methods presented in~\cite{Kareem2025} and~\cite{Shuowen2026HowManyBeams} to the solution.
\end{remark}

\section{Exploration-Centric BCRB Optimization} \label{sec:proposed-P2}
This section considers the exploration-centric variant $\mathsf{(P_2)}$ of the BCRB optimization problem, as given in~(\ref{eqn:bcrb-var}). This variant optimizes the BCRB metric, while requiring the transmit beamformers to be mutually orthogonal. Hence, the optimal solution is guaranteed to explore multiple orthogonal beamforming directions, which makes it suitable for an initial exploration phase. To solve the optimization problem $\mathsf{(P_2)}$, we proceed in a similar way as our approach to $\mathsf{(P_1)}$ in Section~\ref{sec:proposed-P1}. In particular, the problem is reformulated in the covariance domain and then approached via alternating optimization, while taking into account the orthogonal constraint on the transmit beamformers. Each subproblem is then analyzed in the Lagrangian dual domain to derive sufficient conditions for global optimality.

The technical steps involved are similar to those followed in Section~\ref{sec:proposed-P1}. We focus here on the key difference with respect to the optimization $\mathsf{(P_1)}$, namely the orthogonal constraint~(\ref{eqn:bcrb-var-1}) on the transmit beamformers. Specifically, when the problem is expressed in terms of the correlation matrix ${\bf R}_{{\bf V}_k} = {\bf V}_k{\bf V}_k^{\mathsf H}$, the orthogonal constraint can be equivalently written as two constraints: 1) $\frac{M_T}{P}{\bf R}_{{\bf V}_k}$ is an orthogonal projection matrix, and 2) $\rank({\bf R}_{{\bf V}_k}) = M_T$. These two constraints are necessary and sufficient for constructing a matrix ${\bf V}_k$ that satisfies the orthogonal constraint~(\ref{eqn:bcrb-var-1}). Hence, following similar steps that lead to~(\ref{eqn:bcrb-corr-sdp}) in Section~\ref{sec:proposed-P1}, the optimization problem $\mathsf{(P_2)}$ can be expressed in the covariance domain as
\begin{subequations} \label{eqn:bcrb-var-corr-sdp}
\begin{align}
\underset{\substack{{\bf R}_{{\bf V}_k},{\bf R}_{{\bf W}_k}, \\ d_1, \ldots, d_{3L}}}{\text{minimize}} \quad &\sum_{\ell = 1}^{3L} d_\ell\\
\text{subject to} \quad &\begin{bmatrix}
	{\bf J}_k({\bf R}_{{\bf V}_k},{\bf R}_{{\bf W}_k}) &  \sqrt{q_\ell}{\bf e}_\ell \\
	\sqrt{q_\ell}{\bf e}_\ell^\sfT &  d_\ell
\end{bmatrix} \succeq 0, \quad \forall \ell, \\
&\frac{M_T}{P}{\bf R}_{{\bf V}_k} \text{ is an orthogonal projection matrix}, \label{eqn:bcrb-var-corr-sdp-2}\\
&\rank({\bf R}_{{\bf V}_k}) = M_T,\,\, {\bf R}_{{\bf V}_k} \succeq 0, \label{eqn:bcrb-var-corr-sdp-3} \\
&{\bf R}_{{\bf W}_k} \text{ is an orthogonal projection matrix},\\
&\rank({\bf R}_{{\bf W}_k}) = M_R, \,\, {\bf R}_{{\bf W}_k} \succeq 0.
\end{align}
\end{subequations}
Notice that the only difference with~(\ref{eqn:bcrb-corr-sdp}) is the additional orthogonal projection constraint~(\ref{eqn:bcrb-var-corr-sdp-2}), whereas in~(\ref{eqn:bcrb-corr-sdp}), only the trace of ${\bf R}_{{\bf V}_k}$ is constrained.

The optimization problem~(\ref{eqn:bcrb-var-corr-sdp}) can be approached via the same alternating optimization algorithm discussed in Section~\ref{sec:proposed-P1}. In particular, for a fixed ${\bf R}_{{\bf V}_k}$, the optimization problem over ${\bf R}_{{\bf W}_k}$ takes the same form as before (i.e., the optimization problem~(\ref{eqn:bcrb-corr-W})), and for a fixed ${\bf R}_{{\bf W}_k}$, the optimization problem over ${\bf R}_{{\bf V}_k}$ can be written as
\begin{subequations} \label{eqn:bcrb-var-V}
\begin{align}
\underset{{\bf R}_{{\bf V}_k}, d_1, \ldots, d_{3L}}{\text{minimize}} \quad &\sum_{\ell = 1}^{3L} d_\ell\\
\text{subject to} \,\,\,\,\quad &\begin{bmatrix}
	{\bf J}_k({\bf R}_{{\bf V}_k},{\bf R}_{{\bf W}_k}) & \sqrt{q_\ell}{\bf e}_\ell \\
	\sqrt{q_\ell}{\bf e}_\ell^\sfT &  d_\ell
\end{bmatrix} \succeq 0, \,\, \forall \ell,\\
&\frac{M_T}{P}{\bf R}_{{\bf V}_k} \text{ is an orthogonal projection matrix}, \\
&\rank({\bf R}_{{\bf V}_k}) = M_T,\,\, {\bf R}_{{\bf V}_k} \succeq 0.
\end{align}
\end{subequations}
Hence, the optimization problem $\mathsf{(P_2)}$ can be approached by iterating between solving~(\ref{eqn:bcrb-corr-W}) and solving~(\ref{eqn:bcrb-var-V}). In fact, the two optimization problems have very similar forms, and hence,~(\ref{eqn:bcrb-var-V}) can be solved in a manner almost identical to the solution described in Section~\ref{sec:dual-W}. In particular, the following theorem and corollary hold, which can be seen as the recast of Theorem~\ref{thm:dual-W} and Corollary~\ref{cor:dual-W} for solving~(\ref{eqn:bcrb-var-V}).
\begin{theorem} \label{thm:dual-var-V}
	The dual problem of~(\ref{eqn:bcrb-var-V}) can be expressed as
	\begin{equation} \label{eqn:bcrb-var-dual-V}
		\underset{{\bf \Lambda}}{\mathrm{max}} \quad 2\trace({\bf \Lambda}{\bf Q}^{1/2}) - \trace\left({\bf \Lambda}^{\mathsf T}{\bf J}_{k-1}^{({\mathsf P})}{\bf \Lambda}\right) - \frac{2P}{M_T}\sum_{i=1}^{M_T} \mu_{i}\left(\tilde{{\bf P}}_{{\bf \Lambda}  {\bf \Lambda}^{\mathsf T}}\right),
	\end{equation}
	where $\mu_i(\cdot)$ denotes the $i$-th largest eigenvalue of the matrix in the argument, and $\tilde{{\bf P}}_{{\bf \Lambda}  {\bf \Lambda}^{\mathsf T}}$ is as given in~(\ref{eqn:P_Lambda-V}). The dual problem~(\ref{eqn:bcrb-var-dual-V}) is unconstrained and convex. Let $\tilde{{\bf \Lambda}}^{*}$ denote the solution of~(\ref{eqn:bcrb-var-dual-V}). If
	\begin{equation} \label{eqn:spectral-var-V}
		\mu_{M_T}\left(\tilde{{\bf P}}_{\tilde{{\bf \Lambda}}^{*}(\tilde{{\bf \Lambda}}^{*})^{\mathsf T}}\right) > \mu_{M_T+1}\left(\tilde{{\bf P}}_{\tilde{{\bf \Lambda}}^{*}(\tilde{{\bf \Lambda}}^{*})^{\mathsf T}}\right),
	\end{equation}
	then
	\begin{equation}  \label{eqn:R_star-var-V}
		{\bf R}_{{\bf V}_k}^{*} = {\bf V}_k^{*}({\bf V}_k^{*})^{\mathsf H}
	\end{equation}
	is the globally optimal solution of~(\ref{eqn:bcrb-var-V}), where
	\begin{equation} \label{eqn:V_star-var}
		{\bf V}_k^{*} = \begin{bmatrix}
			{\bf v}_{k,1}^{*} & \cdots & {\bf v}_{k,M_T}^{*}
		\end{bmatrix}
	\end{equation}
	and ${\bf v}_{k,1}^{*},\ldots, {\bf v}_{k,M_T}^{*}$ are the eigenvectors of $\tilde{{\bf P}}_{\tilde{{\bf \Lambda}}^{*}(\tilde{{\bf \Lambda}}^{*})^{\mathsf T}}$ (with unit norm) corresponding to the $M_T$ largest eigenvalues.
\end{theorem}

\begin{proof}
The proof is analogous to the proof of Theorem~\ref{thm:dual-W}.
\end{proof}

\begin{corollary} \label{cor:dual-var-V}
	If condition~(\ref{eqn:spectral-var-V}) holds, the semidefinite relaxation of the optimization problem~(\ref{eqn:bcrb-var-V}), as given below, is tight:
	\begin{subequations} \label{eqn:bcrb-var-V-sdr}
		\vspace*{-1em}
		\begin{align}
			\underset{{\bf R}_{{\bf V}_k}, d_1, \ldots, d_{3L}}{\text{minimize}} \quad &\sum_{\ell = 1}^{3L} d_\ell\\
			\text{subject to} \,\,\,\,\quad &\begin{bmatrix}
				{\bf J}_k({\bf R}_{{\bf V}_k},{\bf R}_{{\bf W}_k}) & \sqrt{q_\ell}{\bf e}_\ell \\
				\sqrt{q_\ell}{\bf e}_\ell^\sfT &  d_\ell
			\end{bmatrix} \succeq 0, \,\, \forall \ell, \\
			&{\bf 0} \preceq {\bf R}_{{\bf V}_k} \preceq \frac{P}{M_T}{\bf I}_{N_T},\\
			&\trace({\bf R}_{{\bf V}_k}) \leq P.
		\end{align}
	\end{subequations}
\end{corollary}

\begin{proof}
The proof is analogous to the proof of Corollary~\ref{cor:dual-W}.
\end{proof}


Theorem~\ref{thm:dual-var-V} and Corollary~\ref{cor:dual-var-V} imply that to solve the optimization problem~(\ref{eqn:bcrb-var-V}), one can find the solution $\tilde{{\bf R}}_{{\bf V}_k}^*$ of the semidefinite relaxation~(\ref{eqn:bcrb-var-V-sdr}), recover the optimal solution $\tilde{{\bf \Lambda}}^*$ of the dual problem as in~(\ref{eqn:opt-dual-variable-V}), take the eigenvalue decomposition of $\tilde{{\bf P}}_{\tilde{{\bf \Lambda}}^{*}(\tilde{{\bf \Lambda}}^{*})^{\mathsf T}}$, and then compute the solution ${\bf R}_{{\bf V}_k}^{*}$ according to~(\ref{eqn:R_star-var-V}). This solution is globally optimal provided that condition~(\ref{eqn:spectral-var-V}) holds.

\section{Proposed Active Sensing Strategy} \label{sec:proposed}
In this section, we describe how the two variants of the BCRB optimization problem can be combined in designing an active sensing strategy for MIMO radar systems. To utilize either of these variants, it is crucial to compute the posterior distribution of the sensing parameters $\bd{\theta}$ given the observed measurements. Hence, we first start by showing how this posterior distribution can be efficiently tracked across the sensing stages, before we proceed to describe the proposed active sensing algorithm.

\subsection{Tracking of the Posterior Distribution}
Solving the BCRB optimization problems~(\ref{eqn:bcrb-corr-W}) and~(\ref{eqn:bcrb-corr-V}) requires the computation of the BFIM ${\bf J}_k({\bf R}_{{\bf V}_k},{\bf R}_{{\bf W}_k})$ given in~(\ref{eqn:fisher})-(\ref{eqn:fisher_prior}). The expectation in these equations is taken over the posterior distribution of $\bd{\theta}$ given the previous measurements ${\bf Y}_{1:k-1}$. To approximate this posterior distribution, one can consider a discretized grid of size $K$ for each of the sensing parameters of $\bd{\theta} = [\bd{\phi}^\sfT, \Re\{\bd{\alpha}^\sfT\}, \Im\{\bd{\alpha}^\sfT\}]^\sfT \in \bR^{3L}$ and track the evaluations of the posterior distribution of each sensing parameter over each grid point. However, this approach requires posterior updates on a grid of size $K^{3L}$ in each sensing stage, which is infeasible in practice, even when sensing a single target (i.e., $L=1$). Instead, we can make the observation that the conditional distribution of the reflection coefficients $\bd{\alpha}$ given the angles $\bd{\phi}$ and the previous measurements ${\bf Y}_{1:k-1}$ is complex Gaussian, which follows from the linearity of the measurement model~(\ref{eqn:model}) with respect to $\bd{\alpha}$. 
Then, only the conditional mean and covariance of $\bd{\alpha}$ need to be tracked, and 
a discretized grid is needed only for the posterior distribution of $\bd{\phi}$. 
This approach involves posterior updates on a grid of size $K^L$, which is feasible for $L=2$ sensing targets.

The mean and the covariance matrix of the conditional distribution of $\bd{\alpha}$ given $\bd{\phi}$ and the previous measurements ${\bf Y}_{1:k}$ can be tracked using Kalman filtering~\cite{Ghaddar2025ICC,Sohrabi2021}. In particular, let $\bd{\mu}_k^{(\bd{\alpha} \cond \bd{\phi})}$ and $\bd{\Sigma}_k^{(\bd{\alpha}\cond \bd{\phi})}$ denote the corresponding mean and covariance matrix respectively. Then, we can write
\begin{align} \label{eqn:kalman}
\bd{\mu}_k^{(\boldsymbol{\alpha} \cond \boldsymbol{\phi})} &= \bd{\mu}_{k-1}^{(\boldsymbol{\alpha} \cond \boldsymbol{\phi})} + \bd{\Sigma}_{k-1}^{(\boldsymbol{\alpha}{\bf y}\cond \boldsymbol{\phi})} 
\big(\bd{\Sigma}_{k-1}^{({\bf y}\cond \boldsymbol{\phi})}\big)^{-1}\left(\mathrm{vec}({\bf Y}_k) - \bd{\mu}_{k-1}^{({\bf y} \cond \boldsymbol{\phi})}\right), \nonumber\\
\bd{\Sigma}_k^{(\boldsymbol{\alpha} \cond \boldsymbol{\phi})} &= \bd{\Sigma}_{k-1}^{(\boldsymbol{\alpha} \cond \boldsymbol{\phi})} - \bd{\Sigma}_{k-1}^{(\boldsymbol{\alpha}{\bf y}\cond \boldsymbol{\phi})} 
\big(\bd{\Sigma}_{k-1}^{({\bf y}\cond \boldsymbol{\phi})}\big)^{-1}\big(\bd{\Sigma}_{k-1}^{(\boldsymbol{\alpha}{\bf y}\cond \boldsymbol{\phi})}\big)^\sfH,
\end{align}
where we have used
\begin{align} \label{eqn:mean_y}
\bd{\mu}_{k-1}^{({\bf y} \cond \boldsymbol{\phi})} &= {\bf C}_{\boldsymbol{\phi},k}\bd{\mu}_{k-1}^{(\boldsymbol{\alpha} \cond \boldsymbol{\phi})},\nonumber \\
\bd{\Sigma}_{k-1}^{(\boldsymbol{\alpha}{\bf y}\cond \boldsymbol{\phi})} &= \bd{\Sigma}_{k-1}^{(\boldsymbol{\alpha}\cond \boldsymbol{\phi})} {\bf C}_{\boldsymbol{\phi},k}^\sfH,  \\
\bd{\Sigma}_{k-1}^{({\bf y}\cond \boldsymbol{\phi})} &= {\bf C}_{\boldsymbol{\phi},k}\bd{\Sigma}_{k-1}^{(\boldsymbol{\alpha}\cond \boldsymbol{\phi})}{\bf C}_{\boldsymbol{\phi},k}^\sfH + {\bf W}_k^\sfH{\bf W}_k, \nonumber
\end{align}
and
\begin{equation}
{\bf C}_{\boldsymbol{\phi},k} = \left({\bf V}_k^\sfT \otimes {\bf W}_k^\sfH\right)\left(\overline{{\bf A}}_{\rm T}(\bd{\phi}) * {\bf A}_{\rm R}(\bd{\phi})\right),
\end{equation}
in which ${\bf A}_{\rm S}(\bd{\phi}) = \begin{bmatrix}
{\bf a}_{\rm S}(\phi_1) & \cdots & {\bf a}_{\rm S}(\phi_L)
\end{bmatrix}$ for ${\rm S} \in \{{\rm R}, {\rm T}\}$, $\otimes$ and $*$ denote the Kronecker product and column-wise Khatri-Rao product of two matrices respectively, and $\overline{{\bf A}}_{\rm T}(\bd{\phi})$ denotes the element-wise conjugate of ${\bf A}_{\rm T}(\bd{\phi})$. Since the reflection coefficients are initially assumed to have a zero-mean complex Gaussian distribution with unit variance, the tracking  of the mean and the covariance matrix is initialized with $\bd{\mu}_0^{(\boldsymbol{\alpha} \cond \bd{\phi})} = {\bf 0}_{L\times 1}$ and $\bd{\Sigma}_0^{(\boldsymbol{\alpha} \cond \boldsymbol{\phi})} = \sigma_\alpha^2{\bf I}_L$ for each $\boldsymbol{\phi}$.

It remains to track the posterior distribution of $\bd{\phi}$ given the previous measurements ${\bf Y}_{1:k-1}$. This can be done by a standard application of Bayes' rule. In particular, if $\pi_k^{(\boldsymbol{\phi})}$ denotes the posterior probability density function of $\boldsymbol{\phi}$ given the first $k$ measurements ${\bf Y}_{1:k}$, then $\pi_k^{(\boldsymbol{\phi})}$ can be computed recursively as follows:
\begin{equation} \label{eqn:posterior}
\pi_k^{(\boldsymbol{\phi})} = \frac{\pi_{k-1}^{(\boldsymbol{\phi})}f({\bf Y}_k \cond \boldsymbol{\phi}, {\bf Y}_{1:k-1})}{\int_{\boldsymbol{\phi}'}\pi_{k-1}^{(\boldsymbol{\phi}')}f({\bf Y}_k \cond \boldsymbol{\phi}', {\bf Y}_{1:k-1}) \mathrm{d} \boldsymbol{\phi}'},
\end{equation}
where the denominator in~(\ref{eqn:posterior}) is a normalization term that does not depend on $\boldsymbol{\phi}$. Note that $f({\bf Y}_k \cond \boldsymbol{\phi}, {\bf Y}_{1:k-1})$ is $\mathcal{CN}\left(\bd{\mu}_{k-1}^{({\bf y} \cond \boldsymbol{\phi})}, \bd{\Sigma}_{k-1}^{({\bf y}\cond \boldsymbol{\phi})}\right)$, where $\bd{\mu}_{k-1}^{({\bf y} \cond \boldsymbol{\phi})}$ and $\bd{\Sigma}_{k-1}^{({\bf y}\cond \boldsymbol{\phi})}$ are as computed in~(\ref{eqn:mean_y}).

\subsection{Proposed Active Sensing Strategy}
The proposed active sensing algorithm uses a total of $T$ sensing stages to estimate the desired sensing parameters $\bd{\theta}$. The algorithm takes as input a parameter $T_{\mathrm{explore}}$, such that $0\leq T_{\mathrm{explore}} \leq T$, which indicates the number of sensing stages that are used for exploration. In particular, in the first $T_{\mathrm{explore}}$ sensing stages, the beamforming matrices are designed by solving the exploration-centric variant $\mathsf{(P_2)}$ of the BCRB optimization problem. As described in Section~\ref{sec:proposed-P2}, this variant can be approached via the alternating optimization of~(\ref{eqn:bcrb-corr-W-sdr}) and~(\ref{eqn:bcrb-var-V-sdr}), which correspond to the semidefinite relaxations of the two problems~(\ref{eqn:bcrb-corr-W}) and~(\ref{eqn:bcrb-var-V}). To get the low-rank solutions, the optimal dual variables are computed according to~(\ref{eqn:opt-dual-variable-W}) and~(\ref{eqn:opt-dual-variable-V}), and the beamforming matrices are set according to~(\ref{eqn:W_star}) and~(\ref{eqn:V_star-var}) respectively. If the conditions~(\ref{eqn:spectral-W}) and~(\ref{eqn:spectral-var-V}) hold, these solutions are globally optimal to each subproblem.


\begin{algorithm}[t]
\caption{Proposed Active Sensing Strategy}
\label{alg:active-sensing}
\begin{algorithmic}[1] 
\REQUIRE Inputs $(N_R, N_T, M_R, M_T, T, T_{\mathrm{explore}}, I_\mathrm{max})$
\STATE Initialize $\pi_0^{(\boldsymbol{\phi})}$, $\bd{\mu}_0^{(\bd{\alpha} \cond \bd{\phi})}$ and $\bd{\Sigma}_0^{(\bd{\alpha} \cond \bd{\phi})}$ for each $\bd{\phi}$
\FOR{$k = 1,\ldots,T$}
\STATE Initialize ${\bf V}_k^{(0)}$ to a random beamforming matrix
\FOR{$i = 1,\,\ldots,\, I_\mathrm{max}$}
\STATE Solve~(\ref{eqn:bcrb-corr-W-sdr}) using $({\bf V}_k^{(i-1)}, \pi_{k-1}^{(\bd{\phi})}, \bd{\mu}_{k-1}^{(\bd{\alpha} \cond \bd{\phi})}, \bd{\Sigma}_{k-1}^{(\bd{\alpha} \cond \bd{\phi})})$
\STATE Compute ${\bf \Lambda}^*$ using~(\ref{eqn:opt-dual-variable-W})
\STATE Set ${\bf W}_k^{(i)}$ according to~(\ref{eqn:W_star})
\IF{$k \leq T_{\mathrm{explore}}$}
\STATE Solve~(\ref{eqn:bcrb-var-V-sdr}) using $({\bf W}_k^{(i)}, \pi_{k-1}^{(\bd{\phi})}, \bd{\mu}_{k-1}^{(\bd{\alpha} \cond \bd{\phi})}, \bd{\Sigma}_{k-1}^{(\bd{\alpha} \cond \bd{\phi})})$
\STATE Compute $\tilde{{\bf \Lambda}}^{*}$ using~(\ref{eqn:opt-dual-variable-V})
\STATE Set ${\bf V}_k^{(i)}$ according to~(\ref{eqn:V_star-var})
\ELSE
\STATE Solve~(\ref{eqn:bcrb-corr-V-sdr}) using $({\bf W}_k^{(i)}, \pi_{k-1}^{(\bd{\phi})}, \bd{\mu}_{k-1}^{(\bd{\alpha} \cond \bd{\phi})}, \bd{\Sigma}_{k-1}^{(\bd{\alpha} \cond \bd{\phi})})$
\STATE Compute $\tilde{{\bf \Lambda}}^{*}$ using~(\ref{eqn:opt-dual-variable-V})
\STATE Set ${\bf V}_k^{(i)}$ according to~(\ref{eqn:V_star})
\ENDIF
\ENDFOR
\STATE Make measurement ${\bf Y}_k$
\STATE Update posterior distribution $\pi_k^{(\bd{\phi})}$ according to~(\ref{eqn:posterior})
\STATE Update $\bd{\mu}_k^{(\bd{\alpha} \cond \bd{\phi})}$ and $\bd{\Sigma}_k^{(\bd{\alpha} \cond \bd{\phi})}$ according to~(\ref{eqn:kalman})
\ENDFOR
\RETURN $\bd{\hat\theta} = \E\left[\bd{\theta} \cond {\bf Y}_{1:T}\right]$.
\end{algorithmic}
\end{algorithm}

In the remaining $T - T_{\mathrm{explore}}$ sensing stages, the beamformers are designed by solving the exploitation-centric variant $\mathsf{(P_1)}$ of the BCRB optimization problem. As described in Section~\ref{sec:proposed-P1}, this can be done via the alternating optimization of~(\ref{eqn:bcrb-corr-W-sdr}) and~(\ref{eqn:bcrb-corr-V-sdr}), which are the semidefinite relaxations of the two problems~(\ref{eqn:bcrb-corr-W}) and~(\ref{eqn:bcrb-corr-V}). Then, the optimal dual variables are computed according to~(\ref{eqn:opt-dual-variable-W}) and~(\ref{eqn:opt-dual-variable-V}), and the beamforming matrices are set according to~(\ref{eqn:W_star}) and~(\ref{eqn:V_star}) respectively. If the conditions~(\ref{eqn:spectral-W}) and~(\ref{eqn:spectral-V}) hold, these choices of beamforming matrices are globally optimal for their respective subproblems. For both variants of the optimization, if the sufficient conditions do not hold, the beamformers are set to random vectors in the optimal eigenspace, as discussed in Sections~\ref{sec:proposed-P1} and~\ref{sec:proposed-P2}. The proposed active sensing strategy is summarized in Algorithm~\ref{alg:active-sensing}.

\section{Numerical Experiments} \label{sec:simulations}
In this section, we evaluate the performance of the proposed active sensing strategy as compared to several existing beamforming strategies in the literature. We consider the following baseline schemes:

\textit{1) Random orthogonal beamforming:} In this baseline, the transmit and receive beamforming matrices for all $T$ sensing stages are generated uniformly at random over the set of matrices with orthogonal columns, with the transmit beamformers being scaled to satisfy the power constraint. As in the proposed sensing strategy, the posterior distribution of the sensing parameters is tracked across the sensing stages, and their final estimates are computed via the MMSE estimator~(\ref{eqn:estimator}). 

\textit{2) MMSE-based beamforming:} In this baseline scheme, the transmit and receive beamformers are steered based on the MMSE estimates of the sensing parameters. In particular, if $\bd{\hat{\theta}}_k = \E\left[\bd{\theta} \cond {\bf Y}_{1:k-1}\right]$ is the MMSE estimate of the sensing parameters in the $k$-th sensing stage, the scheme computes an estimate of the target response matrix ${\bf H}(\bd{\hat{\theta}}_k)$, sets the receive beamforming matrix ${\bf W}_k$ to its left singular vectors corresponding to the largest $M_R$ singular values, and sets the transmit beamforming matrix ${\bf V}_k$ to the right singular vector corresponding to the largest singular value, scaled appropriately to satisfy the power constraint. Note that this choice of the beamforming matrices maximizes the beamforming gain if the target response matrix were to be exactly equal to ${\bf H}(\bd{\hat{\theta}}_k)$.


\textit{3) Deep learning solutions:} Several deep-learning-based solutions have been proposed in the literature in order to adaptively design the beamforming matrices for radar sensing. In particular, the approach taken in~\cite{Sohrabi2022} employs a long short-term memory (LSTM) to model the temporal correlation between the received symbols. At each sensing stage, a feedforward neural network is used to map the LSTM hidden state to the set of beamformers that are used in the next sensing stage. The model is trained \emph{end-to-end} for a predetermined number of sensing stages. We note that while the LSTM model in~\cite{Sohrabi2022} is used to design the receive beamformers only (i.e., a single-antenna transmitter is assumed), a similar approach can be used to design both the transmit and receive beamformers.


In the following, we conduct several experiments to compare these baseline schemes with the proposed active sensing strategy. For all the experiments, the angles are assumed to be uniformly distributed over $[\phi_{\mathrm{min}}, \phi_{\mathrm{max}}] = [-\frac{\pi}{3}, \frac{\pi}{3}]$, and a grid set of size $K=1024$ is used for the computation of the posterior distribution. The reflection coefficients are assumed to have unit variance, i.e., $\sigma_\alpha^2=1$. The goal in these experiments is to estimate the angles only (not the reflection coefficients), and hence, the weighting matrix is set to ${\bf Q} = \frac{1}{L}\mathrm{diag}\left(\begin{bmatrix}
	{\bf 1}_{1\times L}, {\bf 0}_{1\times 2L}
\end{bmatrix}\right)$. In this case, the performance of the beamforming strategies is measured by the root mean-squared error (RMSE) given by
\begin{equation}
	\mathrm{RMSE} \triangleq \sqrt{\frac{1}{L}\E\left[\|\bd{\phi} - \bd{\hat \phi}\|^2\right]},
\end{equation}
where $\bd{\phi} = (\phi_1, \ldots, \phi_L)$ are the angles.


\subsection{Single Target: Optimization of Receive Beamformers} \label{sec:simulation-receive}
In the first experiment, we compare the performance of the beamforming strategies for a MIMO radar equipped with a single transmit antenna and $N_R = 8$ receive antennas. The radar aims to sense a single target (i.e., $L=1$) in a total of $T=8$ sensing stages and uses $M_R = 4$ receive beamformers in each sensing stage. We note that since the radar is equipped with a single transmit antenna (i.e., $N_T=M_T=1$), the two variants of the BCRB optimization problems $\mathsf{(P_1)}$ and $\mathsf{(P_2)}$ are equivalent, and hence the parameter $T_{\mathrm{explore}}$ does not affect the optimization in this case. Moreover, since the optimization is only over receive beamformers, we can use a single iteration of optimization (i.e., $I_{\mathrm{max}} = 1$).

Fig.~\ref{fig:mse_1} shows the plot of the WMSE for the different beamforming strategies as a function of the received $\mathrm{SNR} \triangleq \abs{\alpha}^2 P$. The proposed active sensing strategy achieves better angle estimation performance compared to the baseline schemes, especially in the high SNR regime. Indeed, the BCRB metric provides a tighter bound on the WMSE at high SNRs compared to the low SNR regime.

To further compare the performance of the proposed sensing strategy with the LSTM approach, Fig.~\ref{fig:mse_vs_T} shows the performance of the strategies across the sensing stages. Two cases of the LSTM design are considered: 1) the model is trained for the exact number of sensing stages that are used in testing (i.e., a new model needs to be trained whenever the number of sensing stages changes), and 2) a single model is trained for $T=14$ sensing stages and tested for different numbers of sensing stages (i.e., there is a mismatch between training and testing). We can see that the perfectly matched LSTM approach 
improves upon BCRB optimization in the low SNR regime and when the number of sensing stages is small. Nonetheless, this design requires training a new model for varying channel conditions and sensing times, which is infeasible for practical deployment. On the other hand, when a single model is trained, 
the performance of the LSTM approach degrades significantly compared to BCRB optimization, which highlights the drawback of the learning-based solution for active sensing.



\begin{figure}[tbp]
\centering
\hspace{-1em}
\includegraphics[width=0.9\columnwidth]{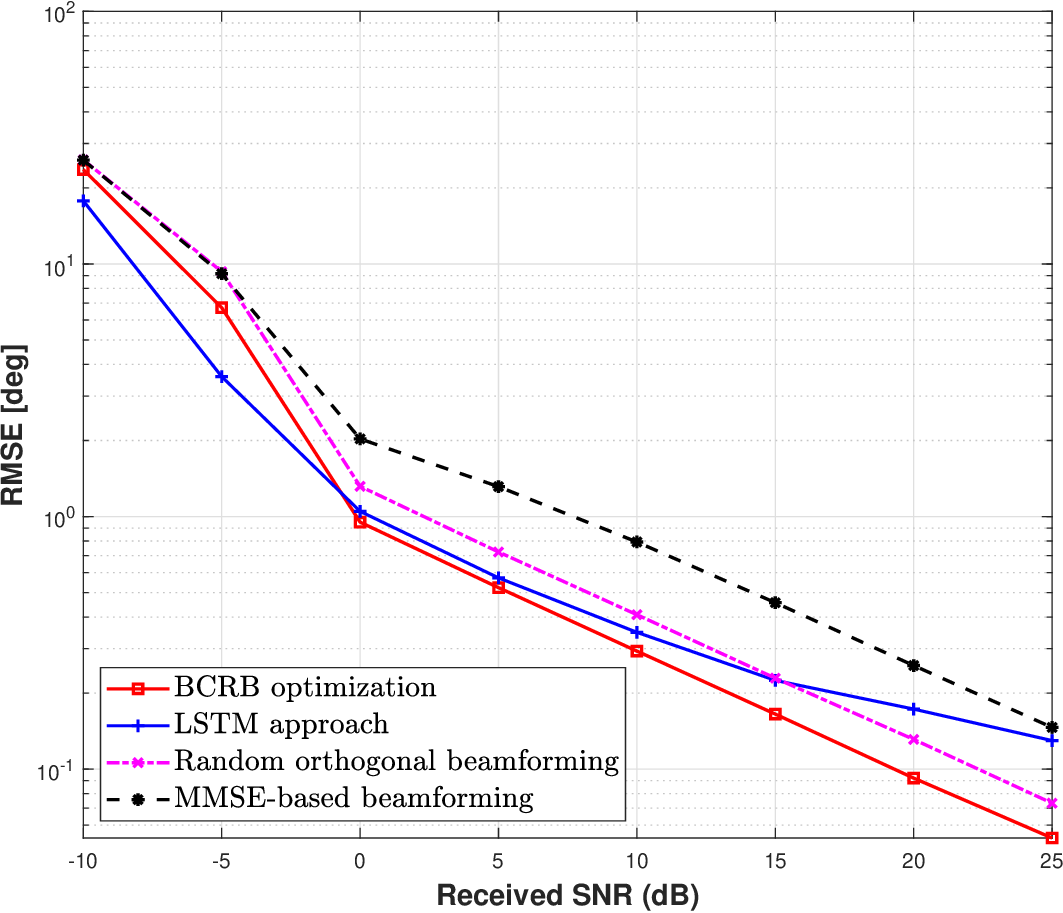}
\caption{Average RMSE versus received SNR for a MIMO sensing system with $N_T=1$ transmit antenna and $N_R=8$ receive antennas that aims to sense a single target using $M_R = 4$ receive beamformers and $T = 8$ sensing stages. Here, the two variants of the BCRB optimization problem are equivalent.}
\label{fig:mse_1}
\end{figure}

\begin{figure}[tbp]
\centering
\hspace{-1em}
\includegraphics[width=0.9\columnwidth]{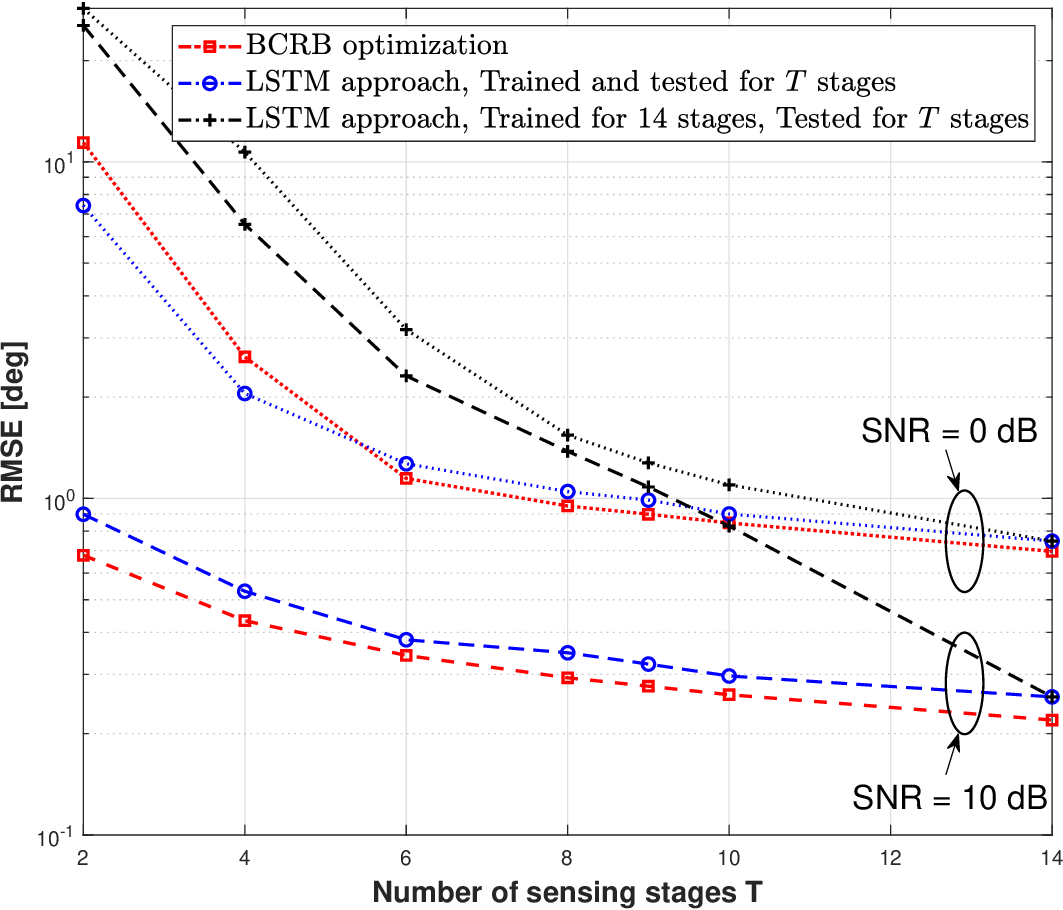}
\caption{Average RMSE versus number of sensing stages for a MIMO sensing system with $N_T=1$ transmit antenna and $N_R=8$ receive antennas that aims to sense a single target using $M_R = 4$ receive beamformers, for two received $\mathrm{SNR}$ levels: $\mathrm{SNR} = 0$ dB, and $\mathrm{SNR} = 10$ dB. Here, the two variants of the BCRB optimization problem are equivalent.}
\label{fig:mse_vs_T}
\end{figure}

\subsection{Single Target: Optimization of Transmit Beamformers} \label{sec:simulation-transmit}
In a second experiment, the beamforming strategies are compared for a MIMO radar equipped with $N_T = 8$ transmit antennas and a single receive antenna. The radar aims to sense a single target in $T=8$ sensing stages, and uses $M_T = 4$ transmit beamformers in each sensing stage. In this case, since $M_T > 1$, the optimization problems $\mathsf{(P_1)}$ and $\mathsf{(P_2)}$ are different, and the choice of the parameter $T_{\mathrm{explore}}$ determines the number of sensing stages in which orthogonal directions are \emph{explored} through transmit beamforming. We note that since the optimization is only over the transmit beamformers, a single iteration of optimization is used (i.e., $I_{\mathrm{max}} = 1$).

We compare three instances of the proposed sensing strategy: 1) $T_{\mathrm{explore}} = 0$, in which case the exploitation-centric optimization problem $\mathsf{(P_1)}$ is considered in all sensing stages, 2) $T_{\mathrm{explore}} = 4$, in which case the exploration-centric optimization $\mathsf{(P_2)}$ is optimized in the first 4 sensing stages and the exploitation-centric $\mathsf{(P_1)}$ is optimized in the remaining 4 sensing stages, and 3) $T_{\mathrm{explore}} = 8$, in which case $\mathsf{(P_2)}$ is optimized in all sensing stages. Fig.~\ref{fig:mse_2} shows the performance of the different beamforming strategies as a function of the received SNR.

Several observations can be made. First, the optimal exploration-exploitation tradeoff in the proposed 
BCRB optimization based active sensing strategy depends on the SNR. In the low-SNR regime, an exploration-only strategy shows the best sensing performance.
In the moderate SNR regime, exploration-followed-by-exploitation performs the best. 
In the high-SNR regime, an exploitation-only strategy has the superior performance, which is closely approached by an exploration-followed-by-exploitation strategy. 
Thus, exploring multiple orthogonal directions through solving the exploration-centric optimization problem $\mathsf{(P_2)}$, particularly in the earlier sensing stages, can improve the sensing performance in the low and moderate SNR regimes. 



\begin{figure}[tbp]
	\centering
	\hspace{-1em}
	\includegraphics[width=0.9\columnwidth]{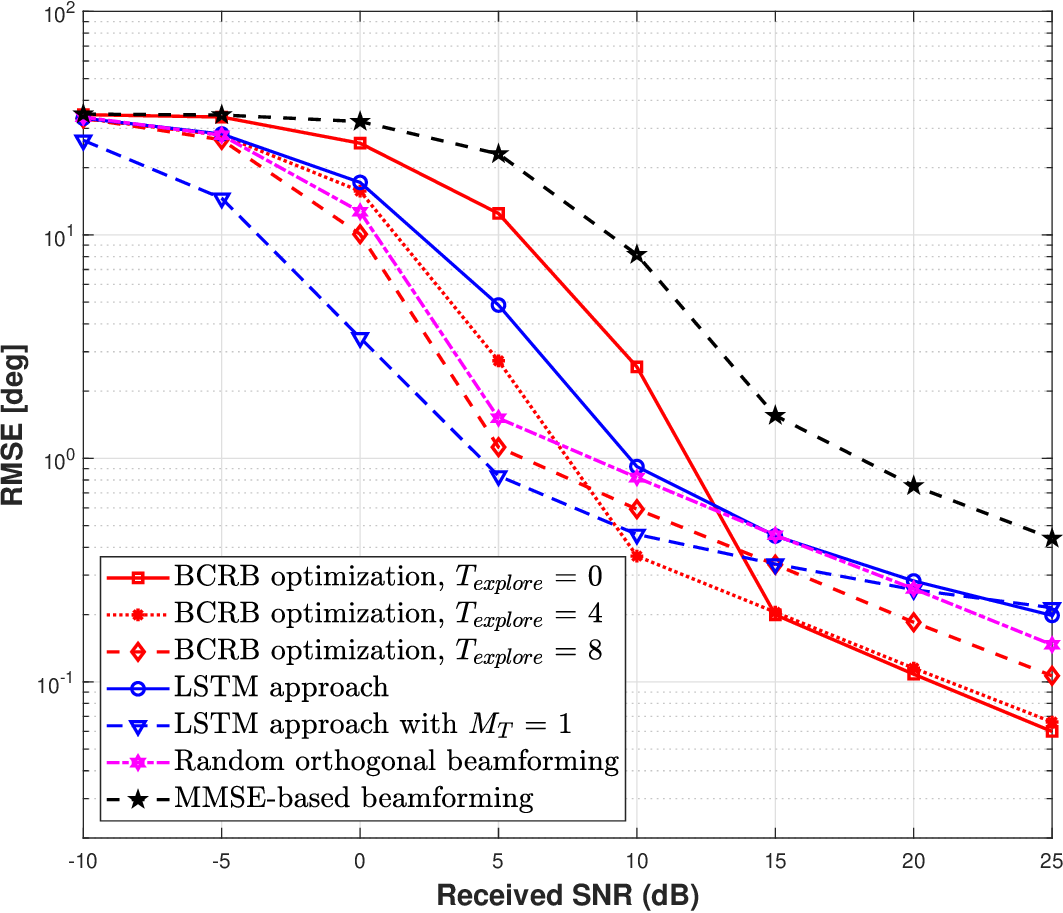}
	\caption{Average RMSE versus received SNR for a MIMO sensing system with $N_T=8$ transmit antennas and $N_R=1$ receive antenna that aims to sense a single target using $M_T = 4$ transmit beamformers and $T = 8$ sensing stages. The parameter $T_{\mathrm{explore}}$ determines the number of sensing stages used for exploration in BCRB optimization.}
	\label{fig:mse_2}
\end{figure}

\begin{figure}[tbp]
	\centering
	\hspace{-1em}
	\includegraphics[width=0.9\columnwidth]{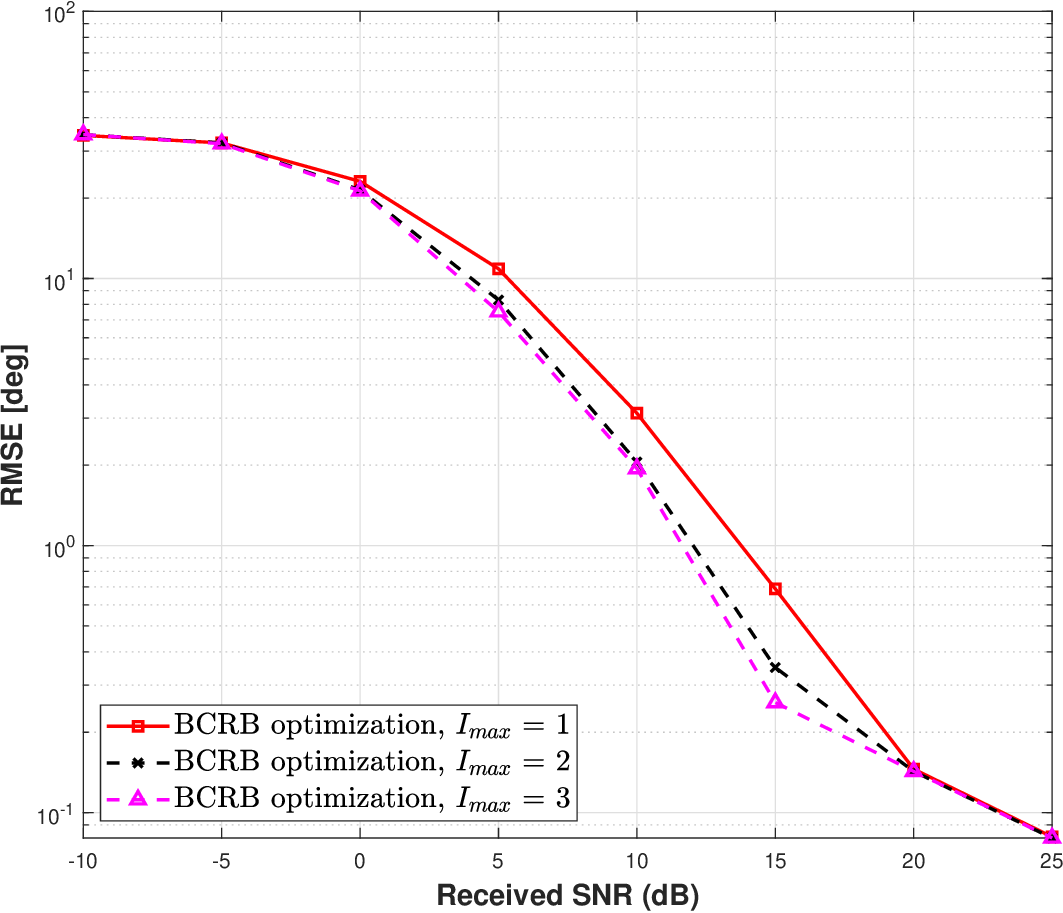}
	\caption{Average RMSE versus received SNR for a MIMO sensing system with $N_T=4$ transmit antennas and $N_R=4$ receive antennas that aims to sense a single target using a single transmit beamformer, a single receive beamformer and $T = 8$ sensing stages. The parameter $I_{\mathrm{max}}$ determines the number of iterations used in the alternating optimization algorithm.}
	\label{fig:mse_3}
\end{figure}

Second, designing multiple orthogonal directions by solving $\mathsf{(P_2)}$ is not the only way to encourage exploration. In fact, a learning-based LSTM model that uses a single transmit beamformer achieves better sensing performance at low SNR compared to both the BCRB-based optimization with orthogonal transmit directions and an LSTM model that designs $M_T=4$ transmit beamformers. 
The LSTM model outperforms BCRB optimization because the BCRB tends to be a looser bound on the WMSE at low SNRs, thus an optimal BCRB metric does not necessarily translate to optimal WMSE performance. 
Furthermore, in contrast to BCRB optimization, the LSTM model is trained end-to-end, and hence it can learn to explore spatial directions in earlier sensing stages, even with a single beamformer. In other words, using additional beamformers (with less power allocated to each beamformer) does not necessarily help the exploration in an end-to-end trained LSTM model. On the other hand, the analytic approach of optimizing the BCRB is naturally greedy and exploitative (i.e., each sensing stage focuses on optimizing a lower bound of the immediate WMSE), and hence exploration needs to be imposed by using multiple beamformers in the optimization formulation $\mathsf{(P_2)}$. 

These observations highlight the gain from considering an exploration phase in multi-stage active sensing, particularly in the low and moderate SNR regimes, and the different ways that such an exploration can be conducted within the active sensing strategy.

\begin{figure*}[t]
	\centering
	\begin{subfigure}{.65\columnwidth}
		\centering
		\hspace*{-1em}
		\includegraphics[scale=0.38]{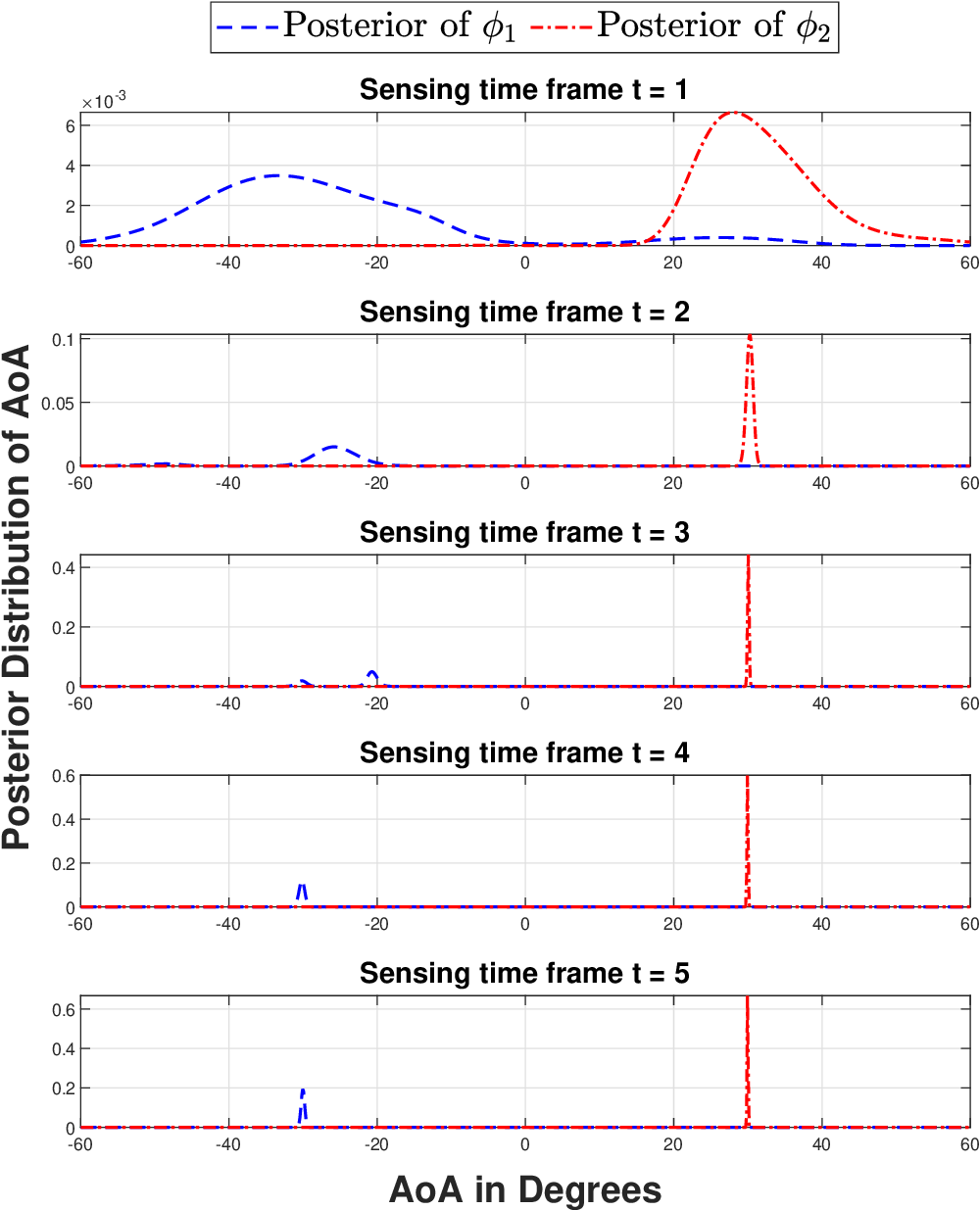}
		\caption{}
		\label{fig:posterior}
	\end{subfigure}%
	\hspace{0.2\columnwidth}%
	\begin{subfigure}{.65\columnwidth}
		\centering
		\hspace*{-1em}
		\includegraphics[scale=0.38]{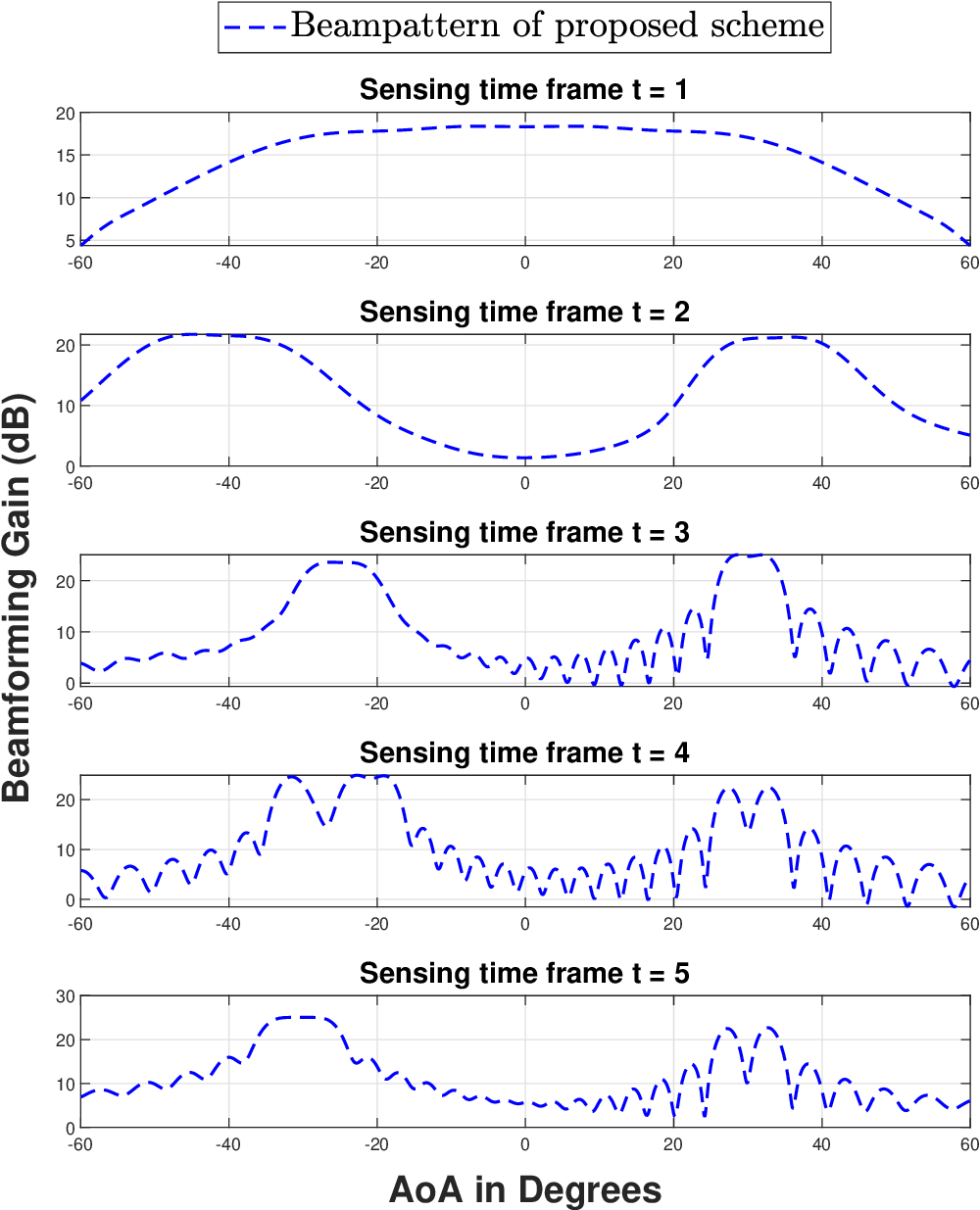}
		\caption{}
		\label{fig:beampattern}
	\end{subfigure}%
	\caption{(a) Marginal posterior distributions of the angles, and (b) beamforming patterns of the beamformers when using the proposed active sensing strategy for a MIMO sensing system model that aims to sense $L = 2$ targets when equipped with a single transmit antenna and $N_R=32$ receive antennas. The system uses $M_R=4$ receive beamformers, $T = 5$ sensing stages and a power level $P = 10$ dB. The true angles are $\phi_1 = -30^{\circ}$ and $\phi_2 = 30^{\circ}$, and the reflection coefficients have $\abs{\alpha_1}^2 = 0.25$ and $\abs{\alpha_2}^2 = 0.75$.}
	\label{fig:beampattern-posterior}
\end{figure*}

\subsection{Single Target: Alternating Optimization} \label{sec:simulation-alternating}
To test the alternating optimization algorithm for designing the transmit and receive beamformers, we simulate the performance of the proposed active sensing strategy for a MIMO radar equipped with $N_T = 4$ transmit antennas and $N_R = 4$ receive antennas, and aims to sense a single target in $T=8$ sensing stages using a single transmit beamformer ($M_T=1$) and a single receive beamformer ($M_R=1$). We run the alternating optimization algorithm for different values of the number of iterations: $I_{\mathrm{max}} = 1$, $I_{\mathrm{max}} = 2$, and $I_{\mathrm{max}} = 3$. We note that since $M_T=1$, the two optimization problems $\mathsf{(P_1)}$ and $\mathsf{(P_2)}$ are equivalent, and hence, the parameter $T_{\mathrm{explore}}$ does not affect the optimization. Fig.~\ref{fig:mse_3} shows the performance as a function of the received SNR. We can see that the sensing performance improves as the number of iterations is increased. This shows the gain from alternating optimization and highlights the fact that the transmit and receive beamformers should be designed jointly for optimal sensing performance.

\subsection{Multiple Targets} \label{sec:simulation-multiple}
To test the performance of the proposed active sensing strategy when there are multiple targets, we consider a MIMO radar that is equipped with a single transmit antenna and $N_R = 32$ receive antennas, and aims to sense two targets (i.e., $L=2$) in $T=5$ sensing stages. The radar uses $M_R=4$ receive beamformers in each sensing stage. Fig.~\ref{fig:beampattern-posterior} plots the beamforming pattern as well as the posterior distribution of the angles across the sensing stages, when the true angles are $\phi_1 = -30^{\circ}$ and $\phi_2 = 30^{\circ}$. We can see that as more measurements are made, the marginal posterior distributions converge to highly concentrated distributions with peaks at the true angles. Moreover, the designed beamformers gradually focus the energy in the direction of the true angles. Notice that the designed beampattern sometimes suppresses the interference from one of the targets, thus facilitating the estimation of the sensing parameters corresponding to the other target. These results demonstrate that the proposed active sensing strategy can be used to sense multiple targets. We note that tracking the two-dimensional posterior distribution of $\phi_1$ and $\phi_2$ is computationally complex.

\section{Conclusion} \label{sec:conclusion}
This paper presents an active sensing framework for designing the transmit and receive beamformers of a MIMO radar system using BCRB optimization. It is shown that the sensing performance of a BCRB-based beamforming strategy is governed by an exploration-exploitation tradeoff, in which the radar has to balance between using a broad set of transmit beamformers in order to explore the sensing environment and using a few focused transmit beamformers in order to exploit the accumulated knowledge about the targets. This work also provides new insights about BCRB optimization. In particular, it is shown that, despite its nonconvexity, the BCRB optimization problem can be solved to global optimality provided that certain sufficient conditions (pertaining to the multiplicity of the eigenvalues of a specific direction matrix) hold. These conditions are derived in terms of the optimal solution of the dual problem, which can be efficiently computed. These theoretical insights provide a comprehensive understanding of BCRB-based active MIMO sensing.

{\appendices
\section{Derivation of BFIM} \label{appx:BFIM}
In this appendix, we derive equations~(\ref{eqn:fisher_data}) and~(\ref{eqn:fisher_prior}). Recall that the $(i,j)$-th entry of the BFIM corresponding to the data measurement model can be expressed as
\begin{equation} 
	\left[{\bf J}_k^{({\mathsf D})}({\bf V}_k,{\bf W}_k)\right]_{i,j} = -\E\left[\frac{\partial^2\log f({\bf Y}_k\cond \bd{\theta}, \bold{Y}_{1:k-1})}{\partial\theta_i \partial\theta_j}\Big| {\bf Y}_{1:k-1}\right].
\end{equation}
From the channel model~(\ref{eqn:model}), it is clear that, conditioned on $(\bd{\theta}, \bold{Y}_{1:k-1})$, the $m$-th column of ${\bf Y}_k$ is distributed as
\begin{equation}
	{\bf y}_{k,m} \cond \bd{\theta}, \bold{Y}_{1:k-1} \, \sim \, \mathcal{CN}\left({\bf W}_{k}^\sfH{\bf H}(\bd{\theta}){\bf v}_{k,m}, \, {\bf W}_{k}^\sfH{\bf W}_{k}\right),
\end{equation}
where ${\bf v}_{k,m}$ is the $m$-th column of ${\bf V}_k$. Hence,
\begin{align}
	&\log f({\bf Y}_k\cond \bd{\theta}, \bold{Y}_{1:k-1}) = \sum_{m=1}^{M_T} \log f({\bf y}_{k,m}\cond \bd{\theta}, \bold{Y}_{1:k-1}) \nonumber \\
	&= c - \sum_{m=1}^{M_T} \Big\|{\bf y}_{k,m} - {\bf W}_{k}^\sfH{\bf H}(\bd{\theta}){\bf v}_{k,m}\Big\|^2_{\left({\bf W}_{k}^\sfH{\bf W}_{k}\right)^{-1}},
\end{align}
where we have used the shorthand notation $\|{\bf x}\|_{{\bf A}}^2 \triangleq {\bf x}^\sfH {\bf A} {\bf x}$, and the constant $c \triangleq -M_R\log \pi - \log\det({\bf W}_{k}^\sfH{\bf W}_{k})$ is independent of $\bd{\theta}$. Thus, the first and second derivatives of the log-density can be expressed as in~(\ref{eqn:logdensity-derivative1}) and~(\ref{eqn:logdensity-derivative2}) at the bottom of the page, in which we have used $\dot{{\bf H}}_i(\bd{\theta})$ and $\ddot{{\bf H}}_{ij}(\bd{\theta})$ as defined in equations~(\ref{eqn:channel-derivative-1}) and~(\ref{eqn:channel-derivative-2}) respectively. By replacing ${\bf y}_{k,m} = {\bf W}_{k}^{\mathsf H}{\bf H}(\bd{\theta}){\bf v}_{k,m} + {\bf W}_k^{\mathsf H}{\bf z}_{k,m}$ in~(\ref{eqn:logdensity-derivative2}) and taking the expectation, we get that the $(i,j)$-th entry of the BFIM corresponding to the data measurement model can be expressed as in~(\ref{eqn:fisher_data}). As a side note, we remark that the BFIM corresponding to the data measurement model can be also expressed in matrix form as
\begin{align}
	&{\bf J}_k^{(\mathsf{D})}({\bf V}_k,{\bf W}_k) \nonumber \\
	&= \sum_{m=1}^{M_T}2\Re\left\{\E\left[\dot{{\bf G}}_{k,m}^\sfH(\bd{\theta}) {\bf R}_{{\bf W}_k} \dot{{\bf G}}_{k,m}(\bd{\theta}) \,\big|\, {\bf Y}_{1:k-1} \right] \right\}, \label{eqn:fisher_data-W-1}
\end{align}
where $\dot{{\bf G}}_{k,m}(\bd{\theta}) = \begin{bmatrix}
	\dot{{\bf H}}_{1}(\bd{\theta}){\bf v}_{k,m} & \cdots & \dot{{\bf H}}_{3L}(\bd{\theta}){\bf v}_{k,m}
\end{bmatrix}$, ${\bf v}_{k,m}$ is the $m$-th beamformer of ${\bf V}_k$, and ${\bf R}_{{\bf W}_k}$ is given in~(\ref{eqn:R_W}). That is, if we evaluate the $(i,j)$-th entry of~(\ref{eqn:fisher_data-W-1}), we get back the expression given in~(\ref{eqn:fisher_data}). This alternative expression of the BFIM will be useful in the proof of Theorem~\ref{thm:dual-W}.
\begin{figure*}[b]
	\hrulefill
	\begin{align}
		\frac{\partial\log f({\bf Y}_k\cond \bd{\theta}, \bold{Y}_{1:k-1})}{\partial\theta_i} &= \sum_{m=1}^{M_T}\Big( {\bf y}_{k,m}^{\mathsf H}\left({\bf W}_{k}^{\mathsf H}{\bf W}_{k}\right)^{-1}{\bf W}_{k}^{\mathsf H}\dot{{\bf H}}_i(\bd{\theta}){\bf v}_{k,m} + {\bf v}_{k,m}^{\mathsf H}\dot{{\bf H}}_i^{\mathsf H}(\bd{\theta}){\bf W}_k\left({\bf W}_{k}^{\mathsf H}{\bf W}_{k}\right)^{-1}{\bf y}_{k,m}  \label{eqn:logdensity-derivative1} \\
		&-{\bf v}_{k,m}^{\mathsf H}\dot{{\bf H}}_i^{\mathsf H}(\bd{\theta}){\bf W}_k\left({\bf W}_{k}^{\mathsf H}{\bf W}_{k}\right)^{-1}{\bf W}_{k}^{\mathsf H}{\bf H}(\bd{\theta}){\bf v}_{k,m} - {\bf v}_{k,m}^{\mathsf H}{\bf H}^{\mathsf H}(\bd{\theta}){\bf W}_k\left({\bf W}_{k}^{\mathsf H}{\bf W}_{k}\right)^{-1}{\bf W}_{k}^{\mathsf H}\dot{{\bf H}}_i(\bd{\theta}){\bf v}_{k,m}\Big) \nonumber
	\end{align}
	\begin{align}
		\frac{\partial^2\log f({\bf Y}_k\cond \bd{\theta}, \bold{Y}_{1:k-1})}{\partial\theta_i\partial \theta_j} &= \sum_{m=1}^{M_T}\Big( {\bf y}_{k,m}^{\mathsf H}\left({\bf W}_{k}^{\mathsf H}{\bf W}_{k}\right)^{-1}{\bf W}_{k}^{\mathsf H}\ddot{{\bf H}}_{ij}(\bd{\theta}){\bf v}_{k,m} + {\bf v}_{k,m}^{\mathsf H}\ddot{{\bf H}}_{ij}^{\mathsf H}(\bd{\theta}){\bf W}_k\left({\bf W}_{k}^{\mathsf H}{\bf W}_{k}\right)^{-1}{\bf y}_{k,m} \nonumber \\
		&-{\bf v}_{k,m}^{\mathsf H}\ddot{{\bf H}}_{ij}^{\mathsf H}(\bd{\theta}){\bf W}_k\left({\bf W}_{k}^{\mathsf H}{\bf W}_{k}\right)^{-1}{\bf W}_{k}^{\mathsf H}{\bf H}(\bd{\theta}){\bf v}_{k,m} - {\bf v}_{k,m}^{\mathsf H}\dot{{\bf H}}_i^{\mathsf H}(\bd{\theta}){\bf W}_k\left({\bf W}_{k}^{\mathsf H}{\bf W}_{k}\right)^{-1}{\bf W}_{k}^{\mathsf H}\dot{{\bf H}}_j(\bd{\theta}){\bf v}_{k,m}  \label{eqn:logdensity-derivative2} \\
		&-{\bf v}_{k,m}^{\mathsf H}\dot{{\bf H}}_j^{\mathsf H}(\bd{\theta}){\bf W}_k\left({\bf W}_{k}^{\mathsf H}{\bf W}_{k}\right)^{-1}{\bf W}_{k}^{\mathsf H}\dot{{\bf H}}_i(\bd{\theta}){\bf v}_{k,m} - {\bf v}_{k,m}^{\mathsf H}{\bf H}^{\mathsf H}(\bd{\theta}){\bf W}_k\left({\bf W}_{k}^{\mathsf H}{\bf W}_{k}\right)^{-1}{\bf W}_{k}^{\mathsf H}\ddot{{\bf H}}_{ij}(\bd{\theta}){\bf v}_{k,m}\Big) \nonumber
	\end{align}
\end{figure*}

Next, we derive the BFIM corresponding to the prior distribution, ${\bf J}_{k-1}^{({\mathsf P})}$. Recall that the $(i,j)$-th entry of this matrix can be expressed as
\begin{equation}
	\left[{\bf J}_{k-1}^{({\mathsf P})}\right]_{i,j} = -\E\left[\frac{\partial^2\log f(\bd{\theta}\cond  \bold{Y}_{1:k-1})}{\partial\theta_i \partial\theta_j}\Big| {\bf Y}_{1:k-1}\right].
\end{equation}
Using Bayes' rule, the posterior distribution of $\bd{\theta}$ given the previous measurements can be written as
\begin{equation}
	f(\bd{\theta}\cond  \bold{Y}_{1:k-1}) = \frac{f(\bd{\theta})\prod\limits_{\ell=1}^{k-1}f({\bf Y}_\ell \cond {\bf Y}_{1:\ell-1}, \bd{\theta})}{f({\bf Y}_{1:k-1})}.
\end{equation}
Therefore, we have
\begin{align}
	\left[{\bf J}_{k-1}^{({\mathsf P})}\right]_{i,j} = &-\sum_{\ell=1}^{k-1}\E\left[\frac{\partial^2\log f({\bf Y}_\ell\cond \bd{\theta}, \bold{Y}_{1:\ell-1})}{\partial\theta_i \partial\theta_j}\Big| {\bf Y}_{1:k-1}\right] \nonumber \\
	&- \E\left[\frac{\partial^2\log f(\bd{\theta})}{\partial\theta_i \partial\theta_j}\Big| {\bf Y}_{1:k-1}\right] \label{eqn:fisher_prior_simplified}
\end{align}
The expectation in the summand of~(\ref{eqn:fisher_prior_simplified}) can be computed by replacing ${\bf Y}_k$ in~(\ref{eqn:logdensity-derivative2})  by ${\bf Y}_\ell$ and taking the expectation. Note that the conditioning in the expectation is on all previous measurements ${\bf Y}_{1:k-1}$, and thus, ${\bf Y}_\ell$, for $\ell = 1,\ldots, k-1$, can be treated as a constant inside such an expectation. Using this observation along with simple algebraic steps, it can be shown that the BFIM corresponding to the prior distribution can be expressed as in~(\ref{eqn:fisher_prior}), which completes the derivation.

	\section{Proof of Theorem~\ref{thm:dual-W}} \label{appx:dual-proof-W}
	The proof of Theorem~\ref{thm:dual-W} is inspired by the proof of~\cite[Theorem 1]{Kareem2025_2}, in which the transmit beamformers of a downlink ISAC system are optimized. Here, the dual problem is derived for the optimization of the receive beamformers in a sensing-only radar system. At a technical level, the approach in~\cite{Kareem2025_2} is to drop the nonconvex constraints in the BCRB optimization problem and derive the Lagrangian dual of the resulting relaxed problem. Here, we derive the Lagrangian dual problem without dropping the nonconvex constraints, and the analysis of this problem allows to derive sufficient conditions under which global optimality is guaranteed (i.e., condition~(\ref{eqn:spectral-W})). We remark that the optimization over the receive beamformers includes an additional nonconvex constraint compared to the optimization over the transmit beamformers, namely the orthogonal projection constraint.

	\subsection{Derivation of the Dual Problem}
	
	Consider the primal problem~(\ref{eqn:bcrb-corr-W}). We derive the Lagrangian dual with respect to the constraint~(\ref{eqn:bcrb-corr-W-1}). Let ${\bf \tilde{\Lambda}}_1, \ldots, {\bf \tilde{\Lambda}}_{3L}$ denote the dual variables defined as
	\begin{equation} \label{eqn:dual-variables}
		{\bf \tilde{\Lambda}}_\ell = \begin{bmatrix}
			{\bf \Gamma}_\ell & -\bd{\lambda}_\ell\\
			-\bd{\lambda}_\ell^\sfT & \nu_\ell
		\end{bmatrix}.
	\end{equation} 
	Then, the dual problem can be written as
	\begin{equation}
		\underset{{\bf \tilde{\Lambda}}_\ell \succeq 0 \, \forall \ell}{\text{max}} \,\, \underset{\substack{{\bf R}_{{\bf W}_k} \in \cR,\\ d_1, \ldots, d_{3L}}}{\min} \sum_{\ell = 1}^{3L}\Big(d_\ell(1-\nu_\ell) +2\sqrt{q_\ell}{\bf e}_\ell^\sfT\bd{\lambda}_\ell -\trace({\bf \Gamma}_\ell {\bf J}_k)\Big),
	\end{equation}
	where $\cR$ denotes the nonconvex constraints~(\ref{eqn:bcrb-corr-W-2}) and~(\ref{eqn:bcrb-corr-W-3}), and ${\bf J}_k \triangleq {\bf J}_k({\bf R}_{{\bf V}_k},{\bf R}_{{\bf W}_k})$ is used for notational convenience. Optimizing over $d_1,\ldots, d_{3L}$, we conclude that $\nu_\ell^{*} = 1$ for all $\ell$. Then, the dual problem becomes
	\begin{equation} \label{eqn:bcrb-dual-2}
		\underset{{\bf \Gamma}_\ell \succeq \bd{\lambda}_\ell\bd{\lambda}_\ell^\sfT \,\forall \ell}{\text{max}} \quad \underset{{\bf R}_{{\bf W}_k} \in \cR}{\text{min}} \quad \sum_{\ell = 1}^{3L}\Big(2\sqrt{q_\ell}{\bf e}_\ell^\sfT\bd{\lambda}_\ell -\trace({\bf \Gamma}_\ell {\bf J}_k)\Big),
	\end{equation}
	where the condition ${\bf \Gamma}_\ell \succeq \bd{\lambda}_\ell\bd{\lambda}_\ell^\sfT$ follows from the Schur complement of the constraint ${\bf \tilde{\Lambda}}_\ell \succeq 0$. Recall that
	\begin{equation} 
		{\bf J}_k({\bf R}_{{\bf V}_k},{\bf R}_{{\bf W}_k}) = {\bf J}_k^{(\mathsf{D})}({\bf R}_{{\bf V}_k},{\bf R}_{{\bf W}_k}) + {\bf J}_{k-1}^{(\mathsf{P})},
	\end{equation}
	hence, solving the inner minimization in~(\ref{eqn:bcrb-dual-2}) is equivalent to solving the following problem
	\begin{equation} \label{eqn:bcrb-inner-min-W}
		\underset{{\bf R}_{{\bf W}_k} \in \cR}{\text{max}} \qquad \trace\left({\bf A}{\bf J}_k^{(\mathsf{D})}({\bf R}_{{\bf V}_k},{\bf R}_{{\bf W}_k})\right),
	\end{equation}
	where ${\bf A} \triangleq \sum_{\ell = 1}^{L} {\bf \Gamma}_\ell \succeq 0$. Note that, for any ${\bf A} \succeq 0$, we have
	\begin{align}
		&\trace\left({\bf A}{\bf J}_k^{(\mathsf{D})}({\bf R}_{{\bf V}_k},{\bf R}_{{\bf W}_k})\right) = \sum_{i,j=1}^{3L}[{\bf A}]_{i,j}\left[{\bf J}_k^{(\mathsf{D})}({\bf R}_{{\bf V}_k},{\bf R}_{{\bf W}_k})\right]_{i,j} \nonumber \\
		&= 2\Re \hspace*{-0.12em} \left\{  \hspace*{-0.12em} \trace  \hspace*{-0.12em} \left( \hspace*{-0.12em} {\bf R}_{{\bf W}_k}  \hspace*{-0.12em} \sum_{i,j=1}^{3L}[{\bf A}]_{i,j}\E\left[\dot{{\bf H}}_i(\bd{\theta}) {\bf R}_{{\bf V}_k} \dot{{\bf H}}_j^\sfH(\bd{\theta}) \big| {\bf Y}_{1:k-1}\right]\hspace*{-0.12em} \right) \hspace*{-0.12em}\right\} \nonumber \\
		&= 2\trace\left({\bf R}_{{\bf W}_k}{\bf P}_{\bf A}\right),
	\end{align}
	where
	\begin{align} \label{eqn:G_A}
		{\bf P}_{{\bf A}} &\triangleq \sum_{i,j=1}^{3L}[{\bf A}]_{i,j}\E\left[\dot{{\bf H}}_i(\bd{\theta}) {\bf R}_{{\bf V}_k} \dot{{\bf H}}_j^\sfH(\bd{\theta}) \big| {\bf Y}_{1:k-1}\right] \\
		&=\sum_{m=1}^{M_T}\E\left[\dot{{\bf G}}_{k,m}(\bd{\theta}){\bf A} \dot{{\bf G}}_{k,m}^{\mathsf H}(\bd{\theta}) \,\big|\, {\bf Y}_{1:k-1}\right],
	\end{align}
	in which $\dot{{\bf G}}_{k,m}(\bd{\theta}) = \begin{bmatrix}
		\dot{{\bf H}}_{1}(\bd{\theta}){\bf v}_{k,m} & \cdots & \dot{{\bf H}}_{3L}(\bd{\theta}){\bf v}_{k,m}
	\end{bmatrix}$ and ${\bf v}_{k,m}$ is the $m$-th beamformer of ${\bf V}_k$. Since ${\bf A} \succeq 0$, we have ${\bf P}_{{\bf A}} \succeq 0$. Thus, the inner minimization in~(\ref{eqn:bcrb-dual-2}) is equivalent to
	\begin{subequations} \label{eqn:bcrb-inner-min-W-2}
		\begin{align}
			\underset{{\bf R}_{{\bf W}_k}}{\text{maximize}} \qquad &\trace\left({\bf R}_{{\bf W}_k}{\bf P}_{\bf A}\right)\\
			\text{subject to} \qquad &{\bf R}_{{\bf W}_k} \text{ is an orthogonal projection matrix},\\
			&\rank({\bf R}_{{\bf W}_k}) = M_R, \,\, {\bf R}_{{\bf W}_k} \succeq 0.
		\end{align}
	\end{subequations}
	The key observation is that (\ref{eqn:bcrb-inner-min-W-2}) has the following analytic solution based on the eigenvalue decomposition of ${\bf P}_{\bf A}$:
	\begin{equation} \label{eqn:R_tilde-W}
		\tilde{{\bf R}}_{\tilde{{\bf W}}_k} = \tilde{{\bf W}}_k\tilde{{\bf W}}_k^\sfH,
	\end{equation}
	where $\tilde{{\bf W}}_k = \begin{bmatrix}
		\tilde{{\bf w}}_1 & \ldots & \tilde{{\bf w}}_{M_R}
	\end{bmatrix}$,
	and $\tilde{{\bf w}}_j$ is the eigenvector of ${\bf P}_{{\bf A}}$ corresponding to the $j$-th largest eigenvalue. The solution~(\ref{eqn:R_tilde-W}) follows from the Ky Fan variational characterization for the sum of the $M_R$ largest eigenvalues of a Hermitian matrix~\cite{Fan1949}. This characterization states that, for any Hermitian matrix ${\bf C} \in \mathbb{C}^{N_R\times N_R}$, the sum of its largest $M_R$ largest eigenvalues can be expressed as
	\begin{equation} \label{eqn:ky-fan}
		\sum_{i=1}^{M_R} \mu_i({\bf C}) = \underset{\substack{{\bf U} \in \bC^{N_R\times M_R}: \\ {\bf U}^\sfH {\bf U} = {\bf I}}}{\max} \,\, \trace\left({\bf U}^\sfH{\bf C} {\bf U} \right),
	\end{equation}
	where the maximizer ${\bf U}^*$ consists of the eigenvectors corresponding to the $M_R$ largest eigenvalues of ${\bf C}$. Note that if we denote ${\bf R}_{{\bf U}} = {\bf UU}^\sfH$, the objective function in~(\ref{eqn:ky-fan}) can be written as $\trace({\bf R}_{{\bf U}}{\bf C})$, and the constraint set in~(\ref{eqn:ky-fan}) can be equivalently expressed as the conditions that: 1) ${\bf R}_{{\bf U}}$ is an orthogonal projection matrix, and 2) $\rank({\bf R}_{{\bf U}}) = M_R$. Hence, this characterization gives the solution of the optimization problem~(\ref{eqn:bcrb-inner-min-W-2}). 
	
	The dual problem~(\ref{eqn:bcrb-dual-2}) can then be written as
	\begin{equation}\label{eqn:bcrb-dual-3}
		\underset{{\bf \Gamma}_\ell \succeq \bd{\lambda}_\ell\bd{\lambda}_\ell^\sfT \,\forall \ell}{\text{maximize}}  \quad \sum_{\ell = 1}^{3L}\Big(2\sqrt{q_\ell}{\bf e}_\ell^\sfT \bd{\lambda}_\ell - \trace({\bf \Gamma}_\ell {\bf J}_{k-1}^{(\mathsf{P})})\Big) - 2 \sum_{i=1}^{M_R} \mu_i\left({\bf P}_{{\bf A}}\right),
	\end{equation}
	where ${\bf A} = \sum_{\ell = 1}^{L} {\bf \Gamma}_\ell \succeq 0$, and $\mu_i(\cdot)$ denotes the $i$-th largest eigenvalue of a matrix. 
	
	Observe that, for any ${\bf \Gamma}_\ell \succeq \bd{\lambda}_\ell\bd{\lambda}_\ell^\sfT$, we have that $\trace({\bf \Gamma}_\ell{\bf J}_{k-1}^{(\mathsf{P})}) \geq \trace(\bd{\lambda}_\ell\bd{\lambda}_\ell^\sfT {\bf J}_{k-1}^{(\mathsf{P})})$ (since ${\bf J}_{k-1}^{(\mathsf{P})}$ is a positive semidefinite matrix). Also, we have ${\bf A} = \sum_{\ell = 1}^{L} {\bf \Gamma}_\ell \succeq \sum_{\ell=1}^{L} \bd{\lambda}_\ell\bd{\lambda}_\ell^\sfT \triangleq {\bf B}$, and hence,
	\begin{equation}
		\sum_{i=1}^{M_R} \mu_i\left({\bf P}_{{\bf A}}\right) \geq \sum_{i=1}^{M_R} \mu_i({\bf P}_{{\bf B}}).
	\end{equation}
	It follows that the optimal solution of~(\ref{eqn:bcrb-dual-3}) should satisfy $\mathbf A = \mathbf B$, so that ${\bf \Gamma}_\ell^{*} = \bd{\lambda}_\ell^{*}(\bd{\lambda}_\ell^{*})^\sfT$ for each $\ell$. Therefore, by denoting ${\bf \Lambda} \triangleq \begin{bmatrix}
		\bd{\lambda}_1 & \cdots & \bd{\lambda}_{3L} 
	\end{bmatrix} \in \mathbb{C}^{3L\times 3L}$, we can write ${\bf B} = {\bf \Lambda} {\bf \Lambda}^\sfT$, and the optimization problem~(\ref{eqn:bcrb-dual-3}) becomes
	\begin{equation} \label{eqn:bcrb-dual-4}
		\underset{{\bf \Lambda} \in \bC^{3L\times 3L}}{\text{maximize}} \quad 2\trace({\bf \Lambda}{\bf Q}^{1/2}) - \trace\left({\bf \Lambda}^\sfT{\bf J}_{k-1}^{(\mathsf{P})}{\bf \Lambda}\right) - 2\sum_{i=1}^{M_R} \mu_i\left({\bf P}_{{\bf \Lambda\Lambda}^\sfT}\right).
	\end{equation}
	The optimization problem~(\ref{eqn:bcrb-dual-4}) is unconstrained, and since it is equivalent to the Lagrangian dual of the primal problem~(\ref{eqn:bcrb-corr-W}), the problem is also convex. As sanity check, to verify the convexity, we check that the objective function in~(\ref{eqn:bcrb-dual-4}) is concave. To see this, notice that the function
	\begin{equation}
		f_1({\bf \Lambda})  \triangleq \trace\left({\bf \Lambda}^\sfT{\bf J}_{k-1}^{(\mathsf{P})}{\bf \Lambda}\right) = \sum_{\ell = 1}^{3L} \bd{\lambda}_\ell^\sfT {\bf J}_{k-1}^{(\mathsf{P})} \bd{\lambda}_\ell
	\end{equation}
	is a convex function (since ${\bf J}_{k-1}^{(\mathsf{P})}$ is a positive semidefinite matrix). Also, consider the function
	\begin{align} \label{eqn:f_2}
		&f_2({\bf \Lambda}) \triangleq  \sum_{i=1}^{M_R} \mu_i\left({\bf P}_{{\bf \Lambda\Lambda}^\sfT}\right) \nonumber \\
		&\stackrel{(a)}{=} \underset{\substack{{\bf U} \in \bC^{N_R\times M_R}: \\ {\bf U}^\sfH {\bf U} = {\bf I}}}{\max} \,\, \trace\left({\bf U}^\sfH{\bf P}_{{\bf \Lambda}{\bf \Lambda}^\sfT} {\bf U} \right) \nonumber \\
		&\stackrel{(b)}{=} \underset{\substack{{\bf U} \in \bC^{N_R\times M_R}: \\ {\bf U}^\sfH {\bf U} = {\bf I}}}{\max} \,\, \sum_{m=1}^{M_T}\E\left[\left\|{\bf U}^\sfH \dot{{\bf G}}_{k,m}(\bd{\theta}) {\bf \Lambda} \right\|_{\mathrm{F}}^2 \,\big|\, {\bf Y}_{1:k-1} \right]
	\end{align}
	where we used in~$(a)$ the Ky Fan variational characterization of the sum of the largest $M_R$ eigenvalues~\cite{Fan1949}, and we used in $(b)$ the definition of ${\bf P}_{{\bf \Lambda}{\bf \Lambda}^\sfT}$. Since the function inside the expectation in~(\ref{eqn:f_2}) is the composition of a convex function and a linear function, and since the pointwise maximum of a family of convex functions is convex, it follows that the function $f_2$ is convex. This confirms that the objective function in~(\ref{eqn:bcrb-dual-4}) is concave, and thus, that the optimization problem is convex.

	\subsection{Sufficient Conditions for Global Optimality}
	Now, we proceed to show that if condition~(\ref{eqn:spectral-W}) holds, then ${\bf R}_{{\bf W}_k}^{*}$ given in~(\ref{eqn:R_star-W}) is the optimal primal solution. Let ${\bf \Lambda}^{*}$ denote the optimal solution of the dual problem~(\ref{eqn:bcrb-dual-W}). Hence, we need to show that $({\bf R}_{{\bf W}_k}^{*}, {\bf \Lambda}^{*})$ is an optimal primal-dual pair of the optimization problem~(\ref{eqn:bcrb-corr-W}). To show this, we utilize the following general conditions for global optimality which are sufficient even when the optimization problem is nonconvex~\cite[Proposition 5.3.2]{Bertsekas2009}. Specifically, it is sufficient to show that
	\begin{enumerate}
		\item[1)] ${\bf R}_{{\bf W}_k}^{*}$ is feasible for the primal problem,
		\item[2)] ${\bf \Lambda}^{*}$ is feasible for the dual problem,
		\item[3)] ${\bf R}_{{\bf W}_k}^{*}$ is the minimizer of the Lagrangian function at ${\bf \Lambda}^{*}$,
		\item[4)] $({\bf R}_{{\bf W}_k}^{*}, {\bf \Lambda}^{*})$ satisfy the complementary slackness condition.
	\end{enumerate}
	
	In the following, we verify that these conditions are satisfied for the pair $({\bf R}_{{\bf W}_k}^{*}, {\bf \Lambda}^{*})$, where ${\bf R}_{{\bf W}_k}^{*}$ is as constructed in~(\ref{eqn:R_star-W}) and ${\bf \Lambda}^{*}$ is the optimal solution of the dual problem~(\ref{eqn:bcrb-dual-W}). First, condition~(1) holds since ${\bf R}_{{\bf W}_k}^{*}$ in~(\ref{eqn:R_star-W}) is an orthogonal projection matrix of rank $M_R$, and hence, is feasible for the primal problem~(\ref{eqn:bcrb-corr-W}). Condition~(2) clearly holds since ${\bf \Lambda}^{*}$ is the solution of the dual problem. Condition~(3) holds by following a similar derivation as the one leading to~(\ref{eqn:R_tilde-W}) while replacing ${\bf A}$ with ${\bf \Lambda}^{*}({\bf \Lambda}^{*})^\sfT$. It remains to show the complementary slackness condition~(4). Recall that at the optimal ${\bf \Lambda}^{*}$, the dual variables in~(\ref{eqn:dual-variables}) can be written as
	\begin{equation}
		{\bf \tilde{\Lambda}}_\ell^{*} = \begin{bmatrix}
			\bd{\lambda}_\ell^* (\bd{\lambda}_\ell^*)^\sfT & -\bd{\lambda}_\ell^*\\
			-(\bd{\lambda}_\ell^*)^\sfT & 1
		\end{bmatrix}, \quad \forall \ell,
	\end{equation}
	where $\bd{\lambda}_\ell^*$ is the $\ell$-th column of ${\bf \Lambda}^*$. Moreover, consider the slack variables
	\begin{equation} \label{eqn:d_star}
		d_\ell^* = q_\ell {\bf e}_\ell^\sfT {\bf J}_k^{-1}( {\bf R}_{{\bf V}_k},{\bf R}_{{\bf W}_k}^*) {\bf e}_\ell,
	\end{equation}
	which satisfy the constraint~(\ref{eqn:bcrb-corr-W-1}) at ${\bf R}_{{\bf W}_k}^*$. Hence, we need to show the complementary slackness condition:
	\begin{equation} \label{eqn:cs}
		\trace\left(\begin{bmatrix}
			{\bf J}_k({\bf R}_{{\bf V}_k},{\bf R}_{{\bf W}_k}^*) & \sqrt{q_\ell}{\bf e}_\ell \\
			\sqrt{q_\ell}{\bf e}_\ell^\sfT &  d_\ell^*
		\end{bmatrix}\begin{bmatrix}
			\bd{\lambda}_\ell^*(\bd{\lambda}_\ell^*)^\sfT & -\bd{\lambda}_\ell^*\\
			-(\bd{\lambda}_\ell^*)^\sfT & 1
		\end{bmatrix}\right)= 0, \,\, \forall \, \ell,
	\end{equation}
	which can be rewritten as the following condition:
	\begin{equation} \label{eqn:cs-2}
		(\bd{\lambda}_\ell^*)^\sfT{\bf J}_k({\bf R}_{{\bf V}_k},{\bf R}_{{\bf W}_k}^*)\bd{\lambda}_\ell^* - 2\sqrt{q_\ell}{\bf e}_\ell^\sfT\bd{\lambda}_\ell^* + d_\ell^* = 0, \quad \forall \, \ell.
	\end{equation}
	This condition holds if we can show that
	\begin{equation} \label{eqn:lambda_star}
		\bd{\lambda}_\ell^* = \sqrt{q_\ell}{\bf J}_k^{-1}({\bf R}_{{\bf V}_k},{\bf R}_{{\bf W}_k}^*){\bf e}_\ell, \quad \forall \, \ell.
	\end{equation}
	In the following, we show that, when condition~(\ref{eqn:spectral-W}) holds, the equality~(\ref{eqn:lambda_star}) follows from the first-order optimality of $\bd{\lambda}_\ell^*$. 
	
	Towards this end, let
	\begin{equation} \label{eqn:f_obj}
		f({\bf \Lambda}) = 2\trace({\bf \Lambda}{\bf Q}^{1/2}) - \trace\left({\bf \Lambda}^{\mathsf T}{\bf J}_{k-1}^{({\mathsf P})}{\bf \Lambda}\right) - 2\sum_{i=1}^{M_R} \mu_{i}\left({\bf P}_{{\bf \Lambda}  {\bf \Lambda}^{\mathsf T}}\right)
	\end{equation}
	denote the objective function in~(\ref{eqn:bcrb-dual-W}). Note that $f$ is concave, and ${\bf \Lambda}^*$ is the maximizer of $f$. Hence, when $f$ is differentiable, the first-order optimality conditions imply that
	\begin{equation} \label{eqn:condition_1}
		\nabla_{\bd{\lambda}_\ell} f({\bf \Lambda}) \Big|_{{\bf \Lambda}  = {\bf \Lambda}^*} = 0, \qquad \forall \, \ell=1,\ldots, 3L.
	\end{equation}
	The main ingredient is to derive conditions under which $f$ is differentiable. Towards this end, let us first focus on the last term in $f$, denoted by
	\begin{equation}
		\tilde{f}({\bf \Lambda}) = \sum_{i=1}^{M_R} \mu_{i}\left({\bf P}_{{\bf \Lambda}  {\bf \Lambda}^{\mathsf T}}\right).
	\end{equation}
	The key observation is that the function $\tilde{f}$ is differentiable if and only if there is strict gap between the $M_R$-th and $(M_R+1)$-th largest eigenvalues, i.e., if and only if $\mu_{M_R}\left({\bf P}_{{\bf \Lambda}  {\bf \Lambda}^{\mathsf T}}\right) > \mu_{M_R+1}\left({\bf P}_{{\bf \Lambda}  {\bf \Lambda}^{\mathsf T}}\right)$. In this case, the gradient is smooth and unique, and the partial derivative of $\tilde{f}({\bf \Lambda})$ with respect to $\Lambda_{j\ell}$ is given by
	\begin{equation} \label{eqn:derivative_1}
		\frac{\partial \tilde{f}({\bf \Lambda})}{\partial \Lambda_{j\ell}} = \trace \left(\left(\frac{\partial \tilde{f}({\bf \Lambda})}{\partial {\bf P}_{{\bf \Lambda\Lambda}^\sfT}}\right)^\sfH \frac{\partial {\bf P}_{{\bf \Lambda\Lambda}^\sfT}}{\partial \Lambda_{j\ell}}\right),
	\end{equation}
	which follows from the chain rule for matrix-valued functions. Under the strict gap condition, we have
	\begin{equation}
		\frac{\partial \tilde{f}({\bf \Lambda})}{\partial {\bf P}_{{\bf \Lambda\Lambda}^\sfT}} = \tilde{{\bf U}}\tilde{{\bf U}}^\sfH,
	\end{equation}
	where $\tilde{{\bf U}} = \begin{bmatrix}
		\tilde{{\bf u}}_1 & \cdots & \tilde{{\bf u}}_{M_R}
	\end{bmatrix}$, and $\tilde{{\bf u}}_i$ is the eigenvector of ${\bf P}_{{\bf \Lambda\Lambda}^\sfT}$ corresponding to the $i$-th largest eigenvalue.
	Hence, at the optimal ${\bf \Lambda}^*$, provided that condition~(\ref{eqn:spectral-W}) holds, this gradient is equal to ${\bf R}_{{\bf W}_k}^*$ given in~(\ref{eqn:R_star-W}). That is, condition~(\ref{eqn:spectral-W}) guarantees the differentiability of $\tilde{f}$ at the optimal ${\bf \Lambda}^*$. Moreover, using standard vector calculus, it can be shown that
	\begin{align}
		\frac{\partial {\bf P}_{{\bf \Lambda\Lambda}^\sfT}}{\partial \Lambda_{j\ell}} = \sum_{m=1}^{M_T} \E\big[&\dot{{\bf G}}_{k,m}(\bd{\theta}) \bd{\lambda}_\ell {\bf e}_{j}^\sfT \dot{{\bf G}}_{k,m}^\sfH(\bd{\theta}) \nonumber \\
		&+ \dot{{\bf G}}_{k,m}(\bd{\theta}){\bf e}_{j}\bd{\lambda}_\ell^\sfT \dot{{\bf G}}_{k,m}^\sfH(\bd{\theta}) \,\big|\, {\bf Y}_{1:k-1} \big].
	\end{align}
	This implies that, at the optimal ${\bf \Lambda}^*$,
	\begin{align}
		&\frac{\partial \tilde{f}({\bf \Lambda})}{\partial \Lambda_{j\ell}} \Bigg|_{{\bf \Lambda}  = {\bf \Lambda}^*} \hspace*{-0.25em}  \nonumber \\
		&= \sum_{m=1}^{M_T}2\Re\left\{{\bf e}_j^\sfT \E\left[\dot{{\bf G}}_{k,m}^\sfH(\bd{\theta}) {\bf R}_{{\bf W}_k}^* \dot{{\bf G}}_{k,m}(\bd{\theta}) \,\big|\, {\bf Y}_{1:k-1} \right] \bd{\lambda}_\ell^*\right\} \nonumber \\
		&= {\bf e}_j^\sfT {\bf J}_k^{(\mathsf{D})}({\bf R}_{{\bf V}_k},{\bf R}_{{\bf W}_k}^*) \bd{\lambda}_\ell^*,
	\end{align}
	where the last equality follows from~(\ref{eqn:fisher_data-W-1}).
	Hence, for each $\ell = 1,\ldots, 3L$, we have
	\begin{align}
		\nabla_{\bd{\lambda}_\ell} f({\bf \Lambda}) \Big|_{{\bf \Lambda} = {\bf \Lambda}^*} \hspace*{-0.25em} &= 2\sqrt{q_\ell}{\bf e}_\ell - 2{\bf J}_{k-1}^{({\mathsf P})}\bd{\lambda}_\ell^* - 2{\bf J}_k^{(\mathsf{D})}({\bf R}_{{\bf V}_k},{\bf R}_{{\bf W}_k}^*) \bd{\lambda}_\ell^* \nonumber \\
		&= 2\sqrt{q_\ell}{\bf e}_\ell  - 2{\bf J}_k({\bf R}_{{\bf V}_k},{\bf R}_{{\bf W}_k}^*) \bd{\lambda}_\ell^*
	\end{align}
	Using the first-order optimality condition~(\ref{eqn:condition_1}), the equality~(\ref{eqn:lambda_star}) follows, which is sufficient to show that the pair $({\bf R}_{{\bf W}_k}^{*}, {\bf \Lambda}^{*})$ satisfies the complementary slackness condition. This shows that the pair $({\bf R}_{{\bf W}_k}^{*}, {\bf \Lambda}^{*})$ satisfies all conditions~(1)-(4). Therefore, if condition~(\ref{eqn:spectral-W}) holds, strong duality is guaranteed for the optimization problem~(\ref{eqn:bcrb-corr-W}), with ${\bf R}_{{\bf W}_k}^{*}$ and ${\bf \Lambda}^{*}$ being the optimal primal and dual solutions that are related via~(\ref{eqn:lambda_star}).
	


	\section{Proof of Theorem~\ref{thm:dual-V}} \label{appx:dual-proof-V}
	The proof of Theorem~\ref{thm:dual-V} follows the same steps as that of Theorem~\ref{thm:dual-W}, with one key distinction. While the optimization problem~(\ref{eqn:bcrb-corr-W}) imposes the orthogonal projection constraint~(\ref{eqn:bcrb-corr-W-2}), the optimization~(\ref{eqn:bcrb-corr-V}) imposes a trace constraint. Following the same steps that lead to~(\ref{eqn:bcrb-inner-min-W}) in Appendix~\ref{appx:dual-proof-W}, it can be shown that the minimization of the Lagrangian function for the optimization problem~(\ref{eqn:bcrb-corr-V}) is equivalent to solving the following optimization problem
	\begin{subequations} \label{eqn:bcrb-inner-min-V}
		\begin{align}
			\underset{{\bf R}_{{\bf V}_k}}{\text{maximize}} \qquad &\trace\left({\bf R}_{{\bf V}_k}\tilde{{\bf P}}_{\bf A}\right) \\ 
			\text{subject to} \qquad &\trace({\bf R}_{{\bf V}_k}) \leq P,\\
			&\rank({\bf R}_{{\bf V}_k}) \leq M_T,\,\, {\bf R}_{{\bf V}_k} \succeq 0,
		\end{align}
	\end{subequations}
	where
	\begin{equation}
		\tilde{{\bf P}}_{{\bf A}} = \sum_{m=1}^{M_R}\E\left[\dot{\tilde{{\bf G}}}_{k,m}(\bd{\theta}){\bf A} \dot{\tilde{{\bf G}}}_{k,m}^{\mathsf H}(\bd{\theta}) \,\big|\, {\bf Y}_{1:k-1}\right],
	\end{equation}
	and $\dot{\tilde{{\bf G}}}_{k,m}(\bd{\theta})$ is as given in Theorem~\ref{thm:dual-V}. The solution to the minimization of the Lagrangian function is thus a rank-1 matrix given by the largest eigenvector of $\tilde{{\bf P}}_{\bf A}$, i.e.,
	\begin{equation} \label{eqn:R_tilde-V}
		\tilde{{\bf R}}_{\tilde{{\bf V}}_k} = P\tilde{{\bf v}}_1\tilde{{\bf v}}_1^\sfH,
	\end{equation}
	where $\tilde{{\bf v}}_1$ is the eigenvector (with unit norm) corresponding to the largest eigenvalue of $\tilde{{\bf P}}_{\bf A}$. The remaining steps to characterize the dual problem follow similarly as in Appendix~\ref{appx:dual-proof-W}, while replacing the sum of eigenvalues in~(\ref{eqn:bcrb-dual-3}) by the single largest eigenvalue of $\tilde{{\bf P}}_{{\bf \Lambda}  {\bf \Lambda}^{\mathsf T}}$, multiplied by $P$. To prove that condition~(\ref{eqn:spectral-V}) is sufficient for finding a globally optimal solution, a similar procedure as in Appendix~\ref{appx:dual-proof-W} is also followed. The key is to show the differentiability of the largest eigenvalue of the direction matrix, which holds if and only if its multiplicity is one (i.e., condition~(\ref{eqn:spectral-V}) holds).

}


\bibliographystyle{IEEEtran}
\bibliography{IEEEabrv,bibliography}

\end{document}